\documentclass[aos, preprint]{imsart}

\RequirePackage[OT1]{fontenc}
\RequirePackage{amsthm,amsmath}
\allowdisplaybreaks[4]
\usepackage{array}
\usepackage{mathrsfs}
\usepackage{amssymb,bbm}
\usepackage{amsfonts}
\usepackage{amsmath,amssymb}
\usepackage{graphicx}
\usepackage{amscd}
\usepackage{color}
\usepackage{float}
\usepackage{mathrsfs}
\usepackage{latexsym,bm}
\usepackage{array}
\usepackage{threeparttable}
\usepackage{enumerate}
\usepackage{rotating}
\usepackage{booktabs}
\usepackage{rotating}
\usepackage{multirow}
\usepackage{subfigure}
\usepackage[plain,noend]{algorithm2e}

\usepackage{diagbox}
\usepackage{natbib}

\usepackage{geometry}
\usepackage{bold-extra}

\DeclareFontFamily{U}{mathx}{\hyphenchar\font45}
\DeclareFontShape{U}{mathx}{m}{n}{
      <5> <6> <7> <8> <9> <10>
      <10.95> <12> <14.4> <17.28> <20.74> <24.88>
      mathx10
      }{}
\DeclareSymbolFont{mathx}{U}{mathx}{m}{n}
\DeclareMathAccent{\widecheck}{\mathalpha}{mathx}{"71}

\DeclareMathAccent{\widecheck}{\mathalpha}{mathx}{"71}

\startlocaldefs \numberwithin{equation}{section}
\theoremstyle{it}
\newtheorem{thm}{Theorem}[section]
\newtheorem{lemma}{Lemma}[section]
\newtheorem{ass}{Assumption}[section]

\newtheorem{pro}{Proposition}[section]
\newtheorem{rem}{Remark}

\newtheorem{alg}{Algorithm}[section]
\endlocaldefs

\def\bftau{\mathbb{\pmb{\tau}}}
\def\ptheta{\mathbb{\pmb{\theta}}}
\begin{document}

\begin{frontmatter}
\title{Adaptive inference for a semiparametric generalized autoregressive conditional heteroskedasticity model}

\begin{aug}
\author{\fnms{Feiyu} \snm{Jiang}\thanksref{m1}\ead[label=e1]{jfy16@mails.tsinghua.edu.cn}},
\author{\fnms{Dong} \snm{Li}\thanksref{m2}\ead[label=e2]{malidong@tsinghua.edu.cn}}
\and
\author{\fnms{Ke} \snm{Zhu}\thanksref{m3,t1}\ead[label=e3]{mazhuke@hku.hk}
\ead[label=u1, url]{http://www.foo.com}}

\affiliation{Tsinghua University\thanksmark{m1}\thanksmark{m2} and University of Hong Kong\thanksmark{m3}}

\thankstext{t1}{Correspondence to: Department of Statistics \& Actuarial Science, University of Hong Kong, Pokfulam Road, Hong Kong.
E-mail address: mazhuke@hku.hk}

\address{
Center for Statistical Science\\
Department of Industry Engineering\\
Tsinghua University\\
Beijing 100084, China\\
\printead{e1}\\
\phantom{E-mail:\ }\printead*{e2}
}

\address
{Department of Statistics\\ \quad \& Actuarial Science\\
University of Hong Kong\\
Hong Kong\\
\printead{e3}
}
\end{aug}

\begin{abstract}
This paper considers a semiparametric generalized autoregressive conditional heteroskedasticity (S-GARCH) model. For this model, we first estimate the time-varying long run component for unconditional variance by the kernel estimator, and then estimate the non-time-varying parameters in GARCH-type short run component by the quasi maximum likelihood estimator (QMLE). We show that the QMLE is asymptotically normal with the parametric convergence rate. Next, we construct a Lagrange multiplier test for linear parameter constraint and a portmanteau test for model checking, and obtain their asymptotic null distributions. Our entire statistical inference procedure works for the non-stationary data with two important features: first, our QMLE and two tests are adaptive to the unknown form of the long run component; second, our QMLE and two tests share the same efficiency and testing power as those in variance targeting method when the S-GARCH model is stationary.

\end{abstract}



\end{frontmatter}

	\section{Introduction}
	
	Since the seminal work of \citet{Engle:1982} and \citet{Bollerslev:1986}, the generalized autoregressive conditional heteroskedasticity (GARCH) model is perhaps the most influential one to capture and forecast the volatility of economic and financial return data.
	However, the GARCH model is often used under the stationarity assumption.
	Due to business cycle, technological progress, preference change and policy switch,
	the underlying structure of data may change over time (see \citet{Hansen:2001}). Hence,
	a non-stationary GARCH model with time-varying parameters
	seems more appropriate to fit the return data in applications;
	see, for example,
	\citet{MS:2004}, \citet{SG:2005},
    \citet{ER:2008},
	\citet{FSR:2008}, \citet{PR:2014},
	\citet{Truquet:2017} and the references therein.
	
	In this paper, we consider a
	semiparametric GARCH (S-GARCH) model of order $(p, q)$
	\begin{flalign}
	\label{semi_model}	y_t=&\sqrt{\tau_t}u_t\text{    with    } \tau_t=\tau(t/T),\\
	\label{garch_model}u_t=&\sqrt{g_t}\eta_t\text{    and    }
	g_t=\omega_0+\sum_{i=1}^q\alpha_{i0} u_{t-i}^2+\sum_{j=1}^{p}\beta_{j0} g_{t-j},
	\end{flalign}
	for $t=1,..., T$, where $\tau(x)$ is a positive smoothing deterministic function with unknown form on the interval $[0,1]$, $u_t$ is a covariance stationary GARCH$(p, q)$ process with
	$\omega_0>0$, $\alpha_{i0}\geq0$ and $\beta_{j0}\geq0$, and $\{\eta_t\}$ is a sequence of independent and identically distributed (i.i.d) random variables with $E\eta_t^2=1$.
	The specification that $\tau_t$ is a function of ratio $t/T$ rather than time $t$ is initiated
	by \citet{Robinson:1989}, and since then, it has become a common scaling scheme in the time series literature;
	see, for example, \citet{DR:2006},
    \citet{CT:2007},
     \citet{XP:2008},
	\citet{ZW:2009}, \citet{ZhangW:2012}, \citet{ZS:2013},
	and \citet{Zhu:2019} to name just a few.
	In (\ref{semi_model})--(\ref{garch_model}), the smooth long run component $\tau_t$ is to depict time-varying parameters in volatility, and the GARCH-type short run component $u_t$ is to capture the temporal dependence.

	By using different specified forms of $\tau(x)$, the S-GARCH model nests
	many often used models, including, for example, the standard GARCH model in \citet{Bollerslev:1986}, the spline-GARCH model in \citet{ER:2008}, and the smooth-transition GARCH model in \citet{AT:2013}.
	The statistical inference for these models has been well studied. However,
	when the specification of $\tau(x)$ is unspecified,  the statistical inference for the S-GARCH model has been less attempted.
	For $p=q=1$, \citet{HL:2010} considered the estimation for the S-GARCH model.
	For $p=0$ (i.e., $\beta_{j0}\equiv0$),
	\citet{PR:2014} constructed a
	score test to
	check the nullity of all $\alpha_{i0}$, and \citet{Truquet:2017} later proposed a projection-based estimation and
	a related Wald test to detect the
	nullity of some of $\alpha_{i0}$. For the general S-GARCH model, the statistical inference methodologies, including estimation, testing and model
	checking, are not available in the literature.

	In this paper, we provide an entire inference procedure for the S-GARCH model to fill this gap. First, we
	give a two-step estimation for the model: the function $\tau(x)$ is estimated by the kernel estimator at step one, and
	the unknown parameter vector in the parametric process $u_t$ is estimated by the quasi maximum likelihood estimator (QMLE) at step two.
	Although the nonparametric estimator at step one has a slower convergence rate, we show that the QMLE at step two is asymptotically normal with a parametric convergence rate. Moreover, we
	construct a new Lagrange multiplier (LM) test for detecting the linear parameter constraint, and
	propose a new portmanteau test for model checking.
	The asymptotic null distributions of
	the LM and portmanteau tests are established. Since our entire inference methodologies allow for unspecified form of $\tau(x)$ and higher
	order $(p, q)$, they alleviate the potential risk of model-misspecification, leading to a broad application scope to handle the non-stationary data.
Finally, we extend the two-step estimation to a multivariate semiparametric BEKK (S-BEKK) model, and
establish the asymptotic normality of the corresponding QMLE.


	Our two-step estimation was previously adopted by \citet{HL:2010} to study the multivariate S-BEKK(1, 1) model.
	For the univariate S-GARCH model, we find a much simpler expression for the asymptotic variance of the QMLE,
	making the related inference methodologies easy-to-implement. Meanwhile,
	we find that the asymptotic variance of the QMLE is adaptive to the
	unknown form of $\tau(x)$. Consequently, the efficiency of the QMLE and the power of its related LM and portmanteau tests
	are invariant regardless of the form of $\tau(x)$.
    However, we can show that the QMLE of the multivariate S-BEKK model no longer enjoys such an adaptiveness feature as in the univariate S-GARCH model.
	Our two-step estimation also shares the similar idea as the variance targeting (VT) estimation
	in \citet{FHZ:2011}, which is only applicable for the stationary S-GARCH model (i.e., $\tau(x)\equiv \tau_0$).
	The difference is that our first step estimator of $\tau(x)$ is non-parametric, while
	the first step estimator of $\tau_0$
	in the VT method is parametric. It turns out that our method requires more involved proof techniques. Interestingly, when the S-GARCH model is stationary,
	our QMLE is asymptotically as efficient as the QMLE in the second step estimation of the VT method,
	although the first step estimator of our method has a slower convergence rate than that of the VT method.
	On the contrary, when the S-GARCH is non-stationary, our QMLE is still valid with the same efficiency as the stationary case due to its adaptiveness feature, while
	the QMLE in the VT method is not applicable any more.

	The remainder of the paper is organized as follows. Section 2 presents the two-step estimation procedure and establishes its related
	asymptotics.
	Section 3 gives a LM test for the linear parameter constraint.
	Section 4 introduces a portmanteau test and obtains its limiting null distribution.
	Section 5 makes a comparison with other estimation methods. Section 6 extends the two-step estimation into the multivariate S-BEKK model.
	Simulation results are reported in Section 7, and applications are given in Section 8.
	Concluding remarks are offered in Section 9.
	Proofs of all theorems are relegated to the Appendix.

	\section{Two-step estimation}
	Let $\theta=(\alpha_{1},...,\alpha_{q},\beta_{1},...,\beta_{p})'\in\Theta$ be the parameter vector in model (\ref{garch_model}), and
	$\theta_0=(\alpha_{10},...,\alpha_{q0},\beta_{10},...,\beta_{p0})'\in\Theta$ be its true value, where
	$\Theta\subset \mathbb{R}_+^{p+q}$ is the parameter space, and $\mathbb{R}_+=(0,\infty)$.
	This section gives a two-step estimation procedure for the S-GARCH model in (\ref{semi_model})--(\ref{garch_model}).
	Our procedure first estimates the nonparametric function $\tau(x)$ in (\ref{semi_model}), and then estimates
	the parameter vector $\theta_0$ in (\ref{garch_model}).
	
	\subsection{Estimation of $\tau(x)$}
	This subsection provides a (Nadaraya-Watson) kernel estimator of $\tau(x)$. To this end, we first need an assumption
	for the identification of $\tau_t$.
	\begin{ass}\label{ident_tau}
		$(\mathrm{i})$  $\sum_{i=1}^{q}\alpha_i+\sum_{j=1}^{p}\beta_j<1$; $(\mathrm{ii})$  $\omega=1-\sum_{i=1}^{q}\alpha_i-\sum_{j=1}^{p}\beta_j$.
	\end{ass}
	
	\noindent Assumption \ref{ident_tau}(i) is equivalent to the covariance stationarity of model (\ref{garch_model}),
	and  Assumption \ref{ident_tau}(ii) is to ensure $Eu_t^2=1$. Under  Assumption \ref{ident_tau}, we have
	$$
	y_t^2=\tau(t/T)+\tau(t/T)(u_t^2-1):=\tau(t/T)+v_t,
	$$	
	where $v_t:=\tau(t/T)(u_t^2-1)$ is a zero-mean process. In other words, $y_t^2$ can be rewritten as a standard non-parametric regression problem with a time-varying mean. Following  Hafner and Linton (2010), it is reasonable to estimate $\tau(x)$ by
	\begin{flalign*}
	\widetilde{\tau}(x)=\frac{\sum_{s=1}^{T}K_h\big(x-\frac{s}{T}\big)y_s^2}{\sum_{s=1}^{T}K_h\big(x-\frac{s}{T}\big)},
	\end{flalign*}
	where $K_h(\cdot)=K(\cdot/h)/h$ with $K(\cdot)$ being a kernel function and $h$ being a bandwidth. Since
	$(1/T)\sum_{s=1}^{T}K_h(x-s/T)=1+O(1/(Th))$ under mild conditions, it is more convenient to estimate
	$\tau(x)$ by
	\begin{flalign}\label{est_tau}
	\widehat{\tau}(x)=\frac{1}{T}{\sum_{s=1}^{T}K_h\Big(x-\frac{s}{T}\Big)y_s^2}.
	\end{flalign}

	To obtain the asymptotic distribution of $\widehat{\tau}(x)$, the following assumptions are needed.
	\begin{ass}\label{ass_tau}
		$(\mathrm{i})$ $\tau:[0,1]\to \mathbb{R}_{+}$ is twice continuously differentiable;
		$(\mathrm{ii})$  $0<\underline{\tau}\leq \inf_{x\in[0,1]}\tau(x)\leq \sup_{x\in[0,1]}\tau(x)\leq \overline{\tau}$, where $\underline{\tau}$ and $\overline{\tau}$ are two positive constants.
	\end{ass}

	\begin{ass}\label{ass_kernel}
		$(\mathrm{i})$ $K:[-1,1]\to \mathbb{R}_+$ is symmetric about zero, bounded and Lipschitz continuous with $\int_{-1}^{1}K(x)dx=1$ and $C_r=\int_{-1}^{1}x^rK(x)dx$; $(\mathrm{ii})$ $h\to0$ and $Th\to\infty$ as $T\to \infty$.
	\end{ass}

	\begin{ass}\label{ass_ut}
		$Eu_{t}^4<\infty$.
	\end{ass}
	
	\noindent Assumption \ref{ass_tau}(i) imposes a smoothness condition on $\tau(x)$, and similar conditions have been
	used in \citet{DR:2006}, \citet{HL:2010}, and \citet{CH:2016}.
	Assumption \ref{ass_tau}(ii) is in line with the condition that
	the intercept term in the standard GARCH model has positive lower and upper bounds.
	Assumption \ref{ass_kernel}(i) holds for many often used kernels, and the bounded support condition on
	$K(x)$ is just to simplify analysis. Assumption \ref{ass_kernel}(ii) requires that $h$ converges to zero at a slower rate than $T^{-1}$,
	and later a more restrictive $h$ is needed for the asymptotics of the estimator of $\theta_0$.
	Assumption \ref{ass_ut} is stronger than Assumption \ref{ident_tau}(i), and it is used to ensure that the
	asymptotic variance of $\widehat{\tau}(x)$ is well defined.

	Let $z_t=u_t^2-1$. The asymptotic normality of $\widehat{\tau}(x)$ is given below.
	
	\begin{thm}\label{thm_kernel}
		Suppose Assumptions \ref{ident_tau}--\ref{ass_ut} hold. Then, for any $x \in (0,1)$,
			$$\sqrt{Th}\big(\widehat{\tau}(x)-\tau(x)-h^2b(x)\big)\to_{\mathcal{L}}N(0,V(x))\mbox{ as }T\to\infty,$$
		where `$\to_{\mathcal{L}}$' stands for the convergence in distribution,
		$$b(x)=\frac{C_2}{2}\frac{\partial^2 \tau(x)}{\partial x^2}\mbox{ and }V(x)=\tau^2(x)\Big\{\int_{-1}^{1} K^2(x)dx\Big\}\sum_{j=-\infty}^{\infty}E(z_tz_{t-j}).$$
	\end{thm}

	Based on $\widehat{\tau}(x)$ in (\ref{est_tau}), we estimate $\tau_t$ by $\widehat{\tau}_t=\widehat{\tau}(t/T)$.
	In practice, $\widehat{\tau}_t$ may have the boundary problem. To circumvent this problem, we follow \citet{CH:2016} to adopt the
	reflection method proposed by \citet{HW:1991}. That is, we generate pseudo data $y_t=y_{-t}$ for $-[Th]\leq t\leq -1$ and $y_t=y_{2T-t}$ for $T+1\leq t\leq T+[Th]$, and then modify $\widehat{\tau}_t$ as
	\begin{flalign}\label{boundary_tau}
	\widehat{\tau}_t=\frac{1}{T}\sum_{s=t-[Th]}^{t+[Th]}K_h\Big(\frac{t-s}{T}\Big)y_s^2.
	\end{flalign}
	Intuitively, the reflection method makes
	the boundary points behave similarly as the interior ones. Similar to \citet{CH:2016},  it can be seen that the reflection method
	gives a bias term of order $O(h^2)$, and hence it does not affect the asymptotics of the estimator of $\theta_0$.
	Although $\widehat{\tau}_t$ in (\ref{boundary_tau}) is used for numerical calculations,
	our proofs below will be based on $\widehat{\tau}_t=\widehat{\tau}(t/T)$ in (\ref{est_tau}) to ease the presentation.
	
	\subsection{Estimation of $\theta_0$}
	
	This subsection considers the QMLE of $\theta_0$.
	Based on Assumption \ref{ident_tau}(ii),  we write the parametric $g_t$ in (\ref{garch_model}) as
	\begin{equation}\label{ideal_g}
	g_t(\theta)=\Big(1-\sum_{i=1}^{q}\alpha_i-\sum_{j=1}^{p}\beta_j\Big)+\sum_{i=1}^{q}\alpha_iu_{t-i}^2+\sum_{j=1}^{p}\beta_jg_{t-j}(\theta).
	\end{equation}
	By assuming that $\eta_{t}\sim N(0, 1)$, the log-likelihood function (multiplied
	by -2 and ignoring constants) of $\{y_t\}$ is
	\begin{equation} \label{ideal_llf}
	L_T(\theta)=\sum_{t=1}^{T}l_t(\theta)\quad \text{with}\quad l_t(\theta)=\frac{u_t^2}{g_t(\theta)}+\log g_t(\theta).
	\end{equation}
	Unfortunately, $L_T(\theta)$ is infeasible for computation, since $\{u_t\}$ are unobservable. Thus, we have to replace
	$\{u_t\}$ by $\{\widehat{u}_{t}\}$, and
	consider the following feasible log-likelihood function
	\begin{equation}\label{f_mle}
	\widehat{L}_T(\theta)=\sum_{t=1}^{T}\widehat{l}_t(\theta)\quad \text{with}\quad \widehat{l}_t(\theta)=\frac{\widehat{u}_t^2}{\widehat{g}_t(\theta)}+\log\widehat{g}_t(\theta),
	\end{equation}
	where $\widehat{u}_{t}=y_{t}/\sqrt{\widehat{\tau}_{t}}$, and $\widehat{g}_t(\theta)$ is computed recursively by
	\begin{equation}\label{hatg}
	\widehat{g}_t(\theta)=\Big(1-\sum_{i=1}^{q}\alpha_i-\sum_{j=1}^{p}\beta_j\Big)
	+\sum_{i=1}^{q}\alpha_i\widehat{u}_{t-i}^2+\sum_{j=1}^{p}\beta_j\widehat{g}_{t-j}(\theta)
	\end{equation}
	with given constant initial values $$\widehat{u}_0=u_0,...,\, \widehat{u}_{1-q}=u_{q-1},\, \widehat{g}_0(\theta)=g_0,...,\, \widehat{g}_{1-p}(\theta)=g_{1-p}.$$
	Based on $\widehat{L}_T(\theta)$ in (\ref{f_mle}), our QMLE of $\theta_0$ is defined as
	$$
	\widehat{\theta}_T=\arg\min_{\theta\in\Theta}\widehat{L}_T(\theta).
	$$
	
	To establish the asymptotics of $\widehat{\theta}_{T}$, denote $\mathcal{A}_{\theta}(z)=\sum_{i=1}^{q}\alpha_iz^i$ and $\mathcal{B}_{\theta}(z)=1-\sum_{i=1}^{p}\beta_iz^i$ with the convention $\mathcal{A}_{\theta}(z)=0$ if $q=0$ and $\mathcal{B}_{\theta}(z)=1$ if $p=0$. The following additional assumptions are imposed.
	
	\begin{ass}\label{ass_garch}
		$(\mathrm{i})$ $\Theta$ is compact;
		$(\mathrm{ii})$ if $p>0$, the polynomials $\mathcal{A}_{\theta_0}(z)$ and $\mathcal{B}_{\theta_0}(z)$  have no common roots, $\mathcal{A}_{\theta_0}(1)\neq 0$, and $\alpha_{q0}+\beta_{p0}\neq 0$; $(\mathrm{iii})$  $\theta_0$ is an interior point of $\Theta$.
	\end{ass}

	\begin{ass}\label{ass_u}
		$E|u_t|^{4(1+\delta_0)}<\infty$ for some $\delta_0>0$.
	\end{ass}

	\begin{ass}\label{ass_eta}
		$(\mathrm{i})$ $\eta_t$ has a continuous and almost surely positive density on $\mathbb{R}$ with $E\eta_t^2=1$; $(\mathrm{ii})$  $E|\eta_t|^{4+4/\delta_0+\delta_1}<\infty$ for some $\delta_1>0$, where  $\delta_0>0$ is defined as in Assumption \ref{ass_u}.
	\end{ass}
	
	\begin{ass}\label{ass_bandwidth}
	 $h=c_hT^{-\lambda_h}$ for some $1/4<\lambda_h<1/2$ and $0<c_h<\infty$.
	\end{ass}

	We offer some remarks on the aforementioned assumptions. Assumption \ref{ass_garch} is regular, and it has been used in
	\citet{HK:2003} and \citet{FZ:2004} to study the QMLE for the stationary GARCH model.
	Assumption \ref{ass_u} is stronger than Assumption \ref{ass_ut}, which is needed for
	the variance target estimator in \citet{FHZ:2011} but not for the QMLE in
	\citet{FZ:2004}.
	Assumption \ref{ass_eta}(i) gives the identification condition for $\theta_0$ based on the QMLE, and
	ensures that the GARCH process $u_t$ is $\beta$-mixing (see \citet{CC:2002}).
	Assumption \ref{ass_eta}(ii) is stronger than the condition $E\eta_t^4<\infty$, which is necessary to derive
	the asymptotic normality of the QMLE for the stationary GARCH model (see \citet{HY:2003}).
	We resort to the stronger conditions of $u_t$ and $\eta_t$ in  Assumptions \ref{ass_u} and \ref{ass_eta}(ii)
	due to the existence of $\tau(x)$ in the S-GARCH model.
	Note that if $\eta_t$ has a light tail (for example, $\eta_t\sim N(0, 1)$), Assumption \ref{ass_eta}(ii) holds for a small value of
	$\delta_0$, and $u_t$ (or the data $y_t$) in Assumption \ref{ass_u} is thus allowed to be heavy-tailed.
	Assumption \ref{ass_bandwidth} requires a more restrictive condition on the bandwidth $h$ than Assumption \ref{ass_kernel}(ii), and
	similar conditions have been adopted by \citet{HL:2010}, \citet{PR:2014}, and \citet{Truquet:2017}. The reason is
	because an undersmoothing $h$ is needed to make the estimation bias from $\widehat{\tau}_t$ negligible
	so that the $\sqrt{T}$-convergence of $\widehat{\theta}_{T}$ holds.

	Denote  $\kappa=E\eta_{t}^{4}$, $g_t=g_{t}(\theta_0)$, $\psi_t=\psi_t(\theta_0)$
	with $\psi_t(\theta)=\{\partial g_t(\theta)/\partial \theta\}/g_t(\theta)$, and
	\begin{flalign}\label{J}
	J_{1}&=E(\psi_t\psi_t'),\,\,\, J_{2}=E(g_t^2)E\big(\psi_t/g_t\big)E\big(\psi_t'/g_t\big).
	\end{flalign}
	Now, we are ready to give the asymptotics  of $\widehat{\theta}_T$ in the following theorem.
	
	\begin{thm}\label{thm_garch}
		Suppose Assumptions \ref{ident_tau}--\ref{ass_kernel}, \ref{ass_garch}(i)--(ii), and \ref{ass_u}--\ref{ass_eta} hold. Then,
		
		$\mathrm{(i)}$ $\widehat{\theta}_{T}\to_{p} \theta_0$ as $T\to\infty$;	
		
		$\mathrm{(ii)}$ furthermore, if Assumption \ref{ass_garch}(iii) holds and Assumption \ref{ass_kernel}(ii) is replaced by Assumption \ref{ass_bandwidth},
		$$
		\sqrt{T}(\widehat{\theta}_T-\theta_0)\to_{\mathcal{L}} N(0,\Sigma)\mbox{ as }T\to\infty,
		$$
		where $\Sigma=(\kappa-1)J_1^{-1}(J_1+J_2)J_{1}^{-1}$, and $J_1$ and $J_2$ are defined in (\ref{J}).
	\end{thm}
	
	\begin{rem}\label{rem_2}
		We can simply estimate $\Sigma$ by its sample version $\widehat{\Sigma}_{T}$, where
		\begin{flalign}\label{est_Sigma}
		\widehat{\Sigma}_{T}=(\widehat{\kappa}_{T}-1)\widehat{J}_{1T}^{-1}(\widehat{J}_{1T}+\widehat{J}_{2T})\widehat{J}_{1T}^{-1}
		\end{flalign}
		with
		\begin{flalign}\label{residual_eta}
		\begin{split}		 \widehat{\kappa}_{T}=\frac{1}{T}\sum_{t=1}^{T}\widehat{\eta}_t^4,\,\,\widehat{J}_{1T}=\frac{1}{T}\sum_{t=1}^{T}\widehat{\psi}_t\widehat{\psi}_t'
		\mbox{ and }\,\,
		\widehat{J}_{2T}=\Big(\frac{1}{T}\sum_{t=1}^{T}\widehat{g}_t^2\Big)\Big(\frac{1}{T}\sum_{t=1}^{T}\frac{\widehat{\psi}_t}{\widehat{g}_t}\Big)
		\Big(\frac{1}{T}\sum_{t=1}^{T}\frac{\widehat{\psi}_t'}{\widehat{g}_t}\Big).
		\end{split}
		\end{flalign}
		Here, $\widehat{\eta}_{t}=\widehat{\eta}_{t}(\widehat{\theta}_{T})$ with $\widehat{\eta}_{t}(\theta)=\widehat{u}_{t}/\sqrt{\widehat{g}_{t}(\theta)}$,
		$\widehat{\psi}_{t}=\widehat{\psi}_{t}(\widehat{\theta}_{T})$ with $\widehat{\psi}_{t}(\theta)=\{\partial \widehat{g}_t(\theta)/\partial \theta\}/\widehat{g}_t(\theta)$, and $\widehat{g}_{t}=\widehat{g}_{t}(\widehat{\theta}_{T})$. Under the conditions of Theorem \ref{thm_garch}, we have $\widehat{\Sigma}_{T}\to_{p}\Sigma$ as $T\to\infty$.
	\end{rem}

	Interestingly, the preceding theorem shows that
	the asymptotic variance of
	$\widehat{\theta}_T$ is independent of $\tau(x)$. Following
	the viewpoint of
	\citet{Robinson:1987}, it means that $\widehat{\theta}_T$ is adaptive to
	the unknown form of $\tau(x)$.
	This adaptiveness feature
	ensures that the efficiency of $\widehat{\theta}_T$ and the power of its related
	tests are unchanged regardless of the form of $\tau(x)$.

	\section{The LM test}
	
	Since \citet{Engle:1982} and \citet{Bollerslev:1986}, testing for the nullity of the parameters in the GARCH model is important in applications.
	This problem can be further generalized to consider the following linear constraint hypothesis
	\begin{equation}\label{null}
	\mathbb{H}_0: R\theta_0=r,
	\end{equation}
	where $R$ is a given $d\times (p+q)$ matrix of rank $d$, and $r$ is a given $d\times 1$ constant vector.
	In this section, we construct a Lagrange multiplier (LM) test statistic $LM_T$ for $\mathbb{H}_0$, where
	$$LM_T=\frac{1}{T}\frac{\partial\widehat{L}_T(\widehat{\theta}_{T|0})}{\partial\theta'}
	\widehat{J}_{1T|0}^{-1}R'\big(R\widehat{\Sigma}_{T|0}R'\big)^{-1}R\widehat{J}_{1T|0}^{-1}
	\frac{\partial\widehat{L}_T(\widehat{\theta}_{T|0})}{\partial\theta}.$$
	Here, $\widehat{\theta}_{T|0}$ is the constrained QMLE of $\theta_0$ under $\mathbb{H}_0$, and
	$\widehat{J}_{1T|0}$ and $\widehat{\Sigma}_{T|0}$ are defined in the same way as
	$\widehat{J}_{1T}$ and $\widehat{\Sigma}_{T}$, respectively, with $\widehat{\theta}_{T}$ replaced by $\widehat{\theta}_{T|0}$.
	The following theorem gives the limiting null distribution of $LM_T$.
	
	\begin{thm}\label{LMtest}
		Suppose the conditions in Theorem \ref{thm_garch}(i) hold, with Assumption \ref{ass_kernel}(ii) replaced by Assumption \ref{ass_bandwidth}. Then, under $\mathbb{H}_0$,
		$$
		LM_T\to_{\mathcal{L}}\chi^2_d\mbox{ as }T\to\infty,
		$$
		where $\chi^2_{d}$ is the chi-square distribution with the degrees of freedom $d$.
	\end{thm}

	Based on Theorem \ref{LMtest}, we can set the rejection region of $LM_{T}$ at level $\alpha$ as
	$\{LM_T>\chi_d^2(\alpha)\},$
	where  $\chi_d^2(\alpha)$ is the  $\alpha$-upper percentile of $\chi_d^2$.
	
	As $\widehat{\theta}_{T}$, our $LM_T$ has the adaptiveness feature, and it has a much broader application scope than the existing LM tests.
	Specifically, the LM test in \citet{Bollerslev:1986} is only applicable for
	the stationary GARCH model, but our $LM_T$ has the superior ability to tackle the non-stationary S-GARCH model.
	For the case of $p=0$, the score test in \citet{PR:2014} can detect the null
	hypothesis that all $\alpha_{i0}$ are zeros, and the Wald test in \citet{Truquet:2017} can
	check the null hypothesis that some of $\alpha_{i0}$ are zeros.
	However, it seems non-trivial to extend these two tests for the general null hypotheses in (3.1), although the score test in
\citet{PR:2014} can be extended to detect the null hypothesis that all $\alpha_{i0}$ and $\beta_{j0}$ are zeros.
	Besides $LM_T$, the Wald and likelihood ratio (LR) tests could also be constructed for $\mathbb{H}_0$.
	When some of $\alpha_{i0}$ or $\beta_{j0}$ are allowed to be zeros under $\mathbb{H}_0$, the Wald and LR tests
	render non-standard limiting null distributions (see \citet{FZ:2009} for general discussions), which have to be simulated
	by the bootstrap method. In contrast, $LM_T$ always has the standard chi-square limiting null distribution, even when
all of the null coefficients are not pinned down in $\mathbb{H}_0$\footnote{Following the arguments in \citet{FZ:2007}, our QMLE $\widehat{\theta}_{T}$ can not be asymptotically normal if $\theta_0$ lies on the boundary of $\Theta$ (i.e.,
 some of $\alpha_{i0}$ or $\beta_{j0}$ are zeros). Since the Wald (including $t$) and LR tests depend on $\widehat{\theta}_{T}$, they can not have the standard chi-square limiting null distribution any more if
$\theta_0$ lies on the boundary of $\Theta$ under $\mathbb{H}_0$. Unlike Wald and LR tests, the limiting distribution of our LM test
$LM_{T}$ depends on the one of
$(RJ_1^{-1}R')^{-1}RJ_1^{-1}\frac{1}{\sqrt{T}}\frac{\partial\widehat{L}_T(\widehat{\theta}_{T|0})}{\partial\theta}$,
which
is always asymptotically normal no matter whether $\theta_0$ lies on the boundary of $\Theta$ or not. Hence, it turns out that
$LM_T$ always has the standard chi-square limiting null distribution. For more discussions on this context, we refer to \citet{Pedersen:2017} and \citet{JLZ:2020}.}.
For practical convenience, we thus only focus on the LM test in this paper, and the consideration of Wald and LR tests is left for future study.
	
	\section{Portmanteau test}

	Since \citet{LB:1978}, the portmanteau test and its variants have been a common tool for checking the model adequacy in time series analysis.
	For the stationary GARCH model,  \citet{LM:1994} proposed a portmanteau test for model checking. However, their test is invalid for the non-stationary S-GARCH model.
	In this section, we follow the idea of \citet{LM:1994} to construct a new portmanteau test to check the adequacy of S-GARCH model,
	and our test seems to be the first formal try in the context of semiparametric time series analysis.

	Let $\widehat{\eta}_{t}$ be the model residual defined as in (\ref{residual_eta}). The idea of our portmanteau test is based
	on the fact that $\{\eta_{t}^2\}$ is a sequence of uncorrelated random variables under
	(\ref{semi_model})--(\ref{garch_model}). Hence, if the S-GARCH model is correctly specified, it is expected that the
	sample autocorrelation function of $\{\widehat{\eta}_{t}^2\}$ at lag $k$, denoted by $\widehat{\rho}_{T,k}$, is
	close to zero, where
	$$\widehat{\rho}_{T,k}=\frac{\sum_{t=k+1}^{T}\big(\widehat{\eta}^2_{t}
		-\overline{\widehat{\eta}^2}\big)\big(\widehat{\eta}^2_{t-k}-\overline{\widehat{\eta}^2}\big)}{\sum_{t=1}^{T}\big(\widehat{\eta}^2_{t}
		-\overline{\widehat{\eta}^2}\big)^2}$$
	with $\overline{\widehat{\eta}^2}$ being the sample mean of $\{\widehat{\eta}_{t}^2\}$.
	Let $\widehat{\rho}_T=(\widehat{\rho}_{T,1},...,\widehat{\rho}_{T,\ell})'$ for an integer $\ell\geq1$,
	and
	\begin{flalign}
	\Sigma_{P1}&=(I_{\ell}, -H,-DJ_1^{-1})\in \mathbb{R}^{\ell\times(\ell+1+p+q)},\label{4_1}\\
	\Sigma_{P2}&=
	\left(\begin{matrix}
	(\kappa-1)I_{\ell}& F & D-FE\big(\frac{\psi_t'}{g_t}\big)\\
	*&Eg_t^2&-Eg_t^2E\big(\frac{\psi_t'}{g_t}\big)\\
	*&*&J_1+J_2
	\end{matrix}\right)\in \mathbb{R}^{(\ell+1+p+q)\times(\ell+1+p+q)}\label{4_2}
	\end{flalign}
	be a symmetric matrix, where $D=(D_1',..., D_\ell')'$ with $D_k=E\{(\eta_{t-k}^2-1)\psi_t'\}$,
	$H=(H_1,..., H_\ell)'$ with $H_k=E\{g_t^{-1}(\eta_{t-k}^2-1)\}$, and $F=(F_1,...,F_{\ell})'$ with $F_k=E\{g_t(\eta_{t-k}^2-1)\}$.
	To facilitate our portmanteau test, we need
	the limiting distribution of $\widehat{\rho}_T$ below.
	
	\begin{thm}\label{thm_port}
		Suppose the conditions in Theorem \ref{thm_garch}(ii) hold. Then, if the S-GARCH model in  (\ref{semi_model})--(\ref{garch_model}) is
		correctly specified,
		$$
		\sqrt{T}\widehat{\rho}_T\to_{\mathcal{L}} N(0,\Sigma_{P})\mbox{ as }T\to\infty,
		$$
		where $\Sigma_{P}=(\kappa-1)^{-1}\Sigma_{P1}\Sigma_{P2}\Sigma_{P1}'$, and $\Sigma_{P1}$ and $\Sigma_{P2}$ are defined in (\ref{4_1})--(\ref{4_2}).
	\end{thm}
	
	As in Remark \ref{rem_2}, $\Sigma_{P}$ can be consistently estimated by its sample version $\widehat{\Sigma}_{P}$. Based on
	$\widehat{\Sigma}_{P}$, our portmanteau test statistic is defined as
	$$Q_{T}(\ell)=T\widehat{\rho}_T'\widehat{\Sigma}_P^{-1}\widehat{\rho}_T.$$
	If the S-GARCH model is correctly specified,   $Q_{T}(\ell)\to_{\mathcal{L}}\chi^2_{\ell}$ as $T\to\infty$ by Theorem \ref{thm_port}.
	So, if the value of $Q_{T}(\ell)$ is larger than $\chi_\ell^2(\alpha)$, the fitted S-GARCH model is inadequate at level $\alpha$. Otherwise, it is adequate at level $\alpha$. In practice, the choice of lag $\ell$ depends on the frequency of the series, and one often chooses
	$\ell$ to be $O(\log(T))$, delivering 6, 9 or 12 for a moderate $T$.
	We shall hightlight that $Q_{T}(\ell)$ also has the adaptiveness feature as $LM_T$, and it is essential to detect the adequacy of the short run GARCH component $u_t$ but not the long run component $\tau_t$, since the form of $\tau_t$ is unspecified in the S-GARCH model.\footnote{To detect whether $\tau_t$ is a constant over time
(i.e., $y_t$ follows a standard GARCH model), one can use the strict stationarity test in \citet{Hong:2017} to check the variance stationarity
of $y_t$.}

	\section{Comparisons with other estimation methods}\label{VT}
	
	This section compares our two-step estimation method with the
	three-step estimation method in \citet{HL:2010} and the
	variance targeting (VT) estimation method in \citet{FHZ:2011}.
	
	\subsection{Comparison with three-step estimation method}
	
	Our two-step estimation method is the same as the first two
	estimation steps in \citet{HL:2010}, where they gave the following asymptotic normality
	result for the S-GARCH($1, 1$) model
	$$\sqrt{T}(\widehat{\theta}_T-\theta_0)\to_{\mathcal{L}} N(0,\Sigma_{\dag})\mbox{ as }T\to\infty,
	$$
	where $\Sigma_{\dag}=J_1^{-1}[(\kappa-1)J_1+J_3+J_4+J_4']J_1^{-1}$ with
	$J_3=(M-E\psi_t)(M-E\psi_t)'Z_1$, $J_4=Z_2(M-E\psi_t)'$,
	\begin{flalign*}
	M=\sum_{j=0}^{\infty}\alpha_{10}\beta_{10}^{j}E\Big(\frac{u^2_{t-j-1}\psi_t}{g_t}\Big),\,\,
	Z_{1}=\sum_{j=-\infty}^{\infty}E(z_tz_{t-j})\,\,\mbox{ and }\,\,Z_{2}=\sum_{j=0}^{\infty}E\big\{z_{t}(\eta_{t-j}^2-1)\psi_{t-j}\big\}.
	\end{flalign*}
	Indeed, we can show that $\Sigma_{\dag}$ and $\Sigma$ are equivalent.
	Since $\Sigma_{\dag}$  involves three infinite summations $M$, $Z_1$ and $Z_2$,
a consistent estimator for $\Sigma_{\dag}$ then involves laborious tuning and smoothing.
On the contrary, our $\Sigma$ has a much simpler expression, and it
	can be directly estimated as shown in Remark \ref{rem_2}.
	
	In \citet{HL:2010}, they further proposed an updated estimator at step three, and claimed
	this updated estimator can achieve the semiparametric efficiency bound when $\eta_t\sim N(0, 1)$.
	Following their idea, we can also update our estimator $\widehat{\theta}_{T}$ to $\widecheck{\theta}_{T}$ at step three.
Specifically, we first update the nonparametric part estimator $\widehat{\tau}_t$ to $\widecheck{\tau}_{t}=\widecheck{\tau}(t/T)$, where
\begin{flalign*}
\widecheck{\tau}(x)=\widehat{\tau}(x)-\Big[\frac{1}{T}\sum_{t=1}^{T}K_h\Big(x-\frac{t}{T}\Big)\frac{\partial^2 \widehat{l}_t(\widehat{\tau},\widehat{\theta}_T)}{\partial\tau^2}\Big]^{-1}\Big[\frac{1}{T}\sum_{t=1}^{T}K_h\Big(x-\frac{t}{T}\Big)\frac{\partial \widehat{l}_t(\widehat{\tau},\widehat{\theta}_T)}{\partial\tau}\Big]
\end{flalign*}
with
$$
\widehat{l}_t({\tau},\widehat{\theta}_T)=\log\widehat{g}_t(\widehat{\theta}_T)+\log(\tau)+\frac{y_t^2}{{\tau}\widehat{g}_t(\widehat{\theta}_T)}.
$$
Then, based on $\widecheck{u}_t^2=y_t^2/\widecheck{\tau}_t$ and some given initial values,  we calculate
$$
\widecheck{g}_t(\theta)=1-\sum_{i=1}^{q}\alpha_i-\sum_{j=1}^{p}\beta_j+\sum_{i=1}^{q}\alpha_i\widecheck{u}_{t-i}^2
+\sum_{j=1}^{p}\beta_j\widecheck{g}_{t-j}(\theta),
$$
and update the parametric part estimator $\widehat{\theta}_{T}$ to $\widecheck{\theta}_{T}$ as follows:
$$
\widecheck{\theta}_T=\widehat{\theta}_T-
\Big[\frac{\partial^2\widecheck{L}_T^*(\widehat{\theta}_T)}{\partial\theta\partial\theta'}\Big]^{-1}
\frac{\partial\widecheck{L}_T^*(\widehat{\theta}_T)}{\partial\theta},
$$
where
\begin{flalign*}
\frac{\partial\widecheck{L}_T^*(\widehat{\theta}_T)}{\partial\theta}=&\sum_{t=1}^{T}\Big[\widecheck{g}_t(\widehat{\theta}_T)^{-1}\frac{\partial\widecheck{g}_t(\widehat{\theta}_T)}{\partial\theta}-\widecheck{G}_T(\widehat{\theta}_T)\Big](1-\widecheck{\eta}_t(\widehat{\theta}_T)^2),\\
\frac{\partial^2\widecheck{L}_T^*(\widehat{\theta}_T)}{\partial\theta\partial\theta'}=&\sum_{t=1}^{T}\Big[\widecheck{g}_t(\widehat{\theta}_T)^{-2}\frac{\partial\widecheck{g}_t(\widehat{\theta}_T)}{\partial\theta}\frac{\partial\widecheck{g}_t(\widehat{\theta}_T)}{\partial\theta'}-\widecheck{G}_T(\widehat{\theta}_T)\widecheck{G}(\widehat{\theta}_T)'\Big]
\end{flalign*}
with $\widecheck{G}_T(\theta)=\frac{1}{T}\sum_{t=1}^{T}\widecheck{g}_t(\theta)^{-1}\frac{\partial\widecheck{g}_t(\theta)}{\partial\theta}$.
	Below, we give the limiting distribution of $\widecheck{\theta}_T$.

	 \begin{thm}\label{thm_improve}
	Suppose the conditions in Theorem \ref{thm_garch}(ii) hold. Then,
	$$
	\sqrt{T}(\widecheck{\theta}_T-\theta_0)\to_{\mathcal{L}} N(0,\Sigma^*)\mbox{ as }T\to\infty,
	$$
	where $\Sigma^{*}=(\kappa-1)J_1^{*-1}(J_1^*+J_2^*)J_{1}^{*-1}$ with \begin{flalign*}
J_1^*=&E\{(\psi_t-E\psi_t)(\psi_t-E\psi_t)'\},\\
J_2^*=&\frac{\omega_0^2}{\gamma_0^2}\Big[Eg_t^{-1}E\psi_t-Eg_t^{-1}\psi_t\Big]\Big[Eg_t^{-1}E\psi_t-Eg_t^{-1}\psi_t\Big]',
	\end{flalign*} and $\omega_0=1-\sum_{i=1}^{q}\alpha_{i0}-\sum_{j=1}^{p}\beta_{j0}$ and $\gamma_0=1-\sum_{j=1}^{p}\beta_{j0}$.
	\end{thm}

	The preceding theorem shows that $\widecheck{\theta}_T$ can not
	achieve the semiparametric efficiency bound as $J_2^*$ is positive definite. Hence, it seems unnecessary to consider the 	third estimation step in \citet{HL:2010}.
	Note that the above updating procedure  was
	also given by \citet{BKRW:1993}, in which they showed the updated estimator can achieve the semiparametric efficiency bound when the data are independent. However, when the data are dependent,
	their conclusion may not be true as demonstrated by Theorem \ref{thm_improve}. The failure of $\widecheck{\theta}_T$
	in our case possibly results from the violation of the following condition
	\begin{flalign}\label{bkrw_cond}
	\frac{1}{\sqrt{T}}\Big\{\frac{\partial\widecheck{L}_T^*(\widehat{\theta}_T)}{\partial\theta}- \frac{\partial{L}_T^*(\widehat{\theta}_T)}{\partial\theta}\Big\}=o_p(1),
	\end{flalign}
	where $\frac{\partial L_T^*(\theta)}{\partial\theta}$ is defined in the same way as $\frac{\partial\widecheck{L}_T^*(\theta)}{\partial\theta}$  with $\widecheck{u}_{t}$ and $\widecheck{g}_{t}({\theta})$ replaced by $u_t$  and ${g}_{t}({\theta})$, respectively.
	In \citet{BKRW:1993}, a  condition  similar to  (\ref{bkrw_cond}) was proved for the independent data. However,
	their technical treatment does not work in our time series setting. This is  because in the updating procedure at step three, the process $\widecheck{g}_t$ utilizes the information before and after time period $t$, so that they are not independent of $\{u_s^2\}_{s\neq t}$.
	
	\subsection{Comparison with VT estimation method}
	
	Our two-step estimation method also
	has a linkage to the VT estimation method in \citet{FHZ:2011}, and this aspect has not been
	explored before.
	The VT method is designed for
	the following covariance stationary GARCH($p, q$) model
	\begin{flalign}
	\label{vt_model}	
	\begin{split}
	&y_t=\sqrt{h_t}\eta_t\\
	&\text{with    } h_t=\tau_0\Big(1-\sum_{i=1}^{q}\alpha_{i0}-\sum_{j=1}^{p}\beta_{j0}\Big)+\sum_{i=1}^q\alpha_{i0} y_{t-i}^2+\sum_{j=1}^{p}\beta_{j0} h_{t-j},
	\end{split}
	\end{flalign}
	where $\tau_0$ is a positive parameter, and $\alpha_{i0}$, $\beta_{j0}$ and $\eta_t$ are defined as before.
	Indeed, model (\ref{vt_model}) is just our stationary S-GARCH model, and it is also
	an alternative reparametrization version of the conventional covariance stationary GARCH model. Since $Ey_t^2=\tau_0$ under model (\ref{vt_model}),
	the VT method first estimates $\tau_0$
	by $\overline{\tau}_{T}$, and then estimates $\theta_0$ by the QMLE $\overline{\theta}_{T}$, where
	\begin{flalign}\label{vt_qmle}
	\begin{split}
	\overline{\tau}_{T}=\frac{1}{T}\sum_{t=1}^{T} y_{t}^2\mbox{ and }
	\overline{\theta}_T=\arg\min_{\theta\in\Theta}\overline{L}_T(\theta)
	\mbox{ with }\overline{L}_T(\theta)=\sum_{t=1}^{T}\frac{\overline{u}_t^2}{\overline{g}_t(\theta)}+\log\overline{g}_t(\theta).
	\end{split}
	\end{flalign}
	Here, $\overline{u}_{t}=y_{t}/\sqrt{\overline{\tau}_T}$, and $\overline{g}_{t}(\theta)$ is defined in the same way as
	$\widehat{g}_t(\theta)$ in (\ref{hatg}) with $\widehat{u}_{t}$ replaced by $\overline{u}_{t}$.
	Clearly, the difference of two methods is that our method estimates the unknown function $\tau(x)$ nonparametrically, while
	the VT method estimates the unknown constant parameter $\tau_0$ by the sample mean of $y_{t}^2$.
	It turns out that two methods require different technical treatments and
	give different application scopes.
	From a statistical point of view, the proof techniques for VT method
	rely on the facts that the objective function $\overline{L}_T(\theta)$ is differential around $\tau_0$ and the first step estimator $\overline{\tau}_T$ is $\sqrt{T}$-consistent. However, neither of these facts holds for our method, and we thus
	need develop new proof techniques based on more restrictive conditions for $u_t$ and $\eta_t$.
	From a practical point of view,  our method works for the either stationary or non-stationary
	S-GARCH model, while the VT method does only for the stationary S-GARCH model.
	Hence, our method has a much broader application scope than the VT method.
	
	By revisiting Theorem 2.1 in \citet{FHZ:2011}, we further find
	that the asymptotic variance of $\overline{\theta}_T$ is the same as the one of $\widehat{\theta}_{T}$ in
	Theorem \ref{thm_garch}.
	That is,  our QMLE $\widehat{\theta}_{T}$ and the QMLE $\overline{\theta}_T$ in the VT method
	have the same asymptotic efficiency, although our first step estimator has a slower convergence rate $\sqrt{Th}$ than the parametric convergence rate $\sqrt{T}$. This novel feature has not been discovered in the literature, and it
	makes our two-step method more attractive than the VT method, since our QMLE does not suffer any efficiency loss for the stationary S-GARCH model, and at the same time, our QMLE can still work with the same efficiency (due to the adaptiveness
	feature) for the non-stationary S-GARCH model. As expected, similar features also hold for our tests $LM_{T}$ and $Q_{T}(\ell)$, and these findings will be
	further illustrated by simulation studies.
	
\section{Extension to multivariate S-BEKK model}\label{sec:BEKK}
In this section, we extend the two-step estimation for the S-GARCH model to the multivariate semiparametric BEKK (S-BEKK) model.
Let $\{\mathbf{y}_t\}_{t=1}^{T}$ be a sequence of random vectors with dimension $N\geq 1$.
Assume $\mathbf{y}_t$ satisfies the following S-BEKK model
\begin{flalign}
\label{ym}\mathbf{y}_t=&\bftau_{t}^{1/2}\mathbf{u}_t~\mbox{with}~\bftau_t=\bftau(t/T),\\
\label{BEKKm}\mathbf{u}_t=&\mathbf{g}_t^{1/2}\pmb{\eta}_t~\mbox{and}~ \mathbf{g}_t=W_0+\sum_{i=1}^{q}A_{i0}\mathbf{u}_{t-i}\mathbf{u}_{t-i}'A_{i0}'+\sum_{j=1}^{p}B_{j0}\mathbf{g}_{t-j}B_{j0}',
\end{flalign}
for $t=1,..., T$, where $\bftau(x)\in\mathbb{R}^{N\times N}$ is a positively definite,  smoothing and  deterministic matrix with unknown form on the interval $[0,1]$, $\mathbf{u}_t$ is a covariance stationary BEKK$(p, q)$ process parameterized by $N\times N$ matrices $A_{i0}$, $i=1,...,q$, $B_{j0}$, $j=1,...,p$ and  $W_0:=I_N-\sum_{i=1}^{q}A_{i0}A_{i0}'-\sum_{j=1}^{p}B_{j0}B_{j0}'$,  and  $\{\pmb{\eta}_t\}$ is a sequence of i.i.d random vectors satisfying $E\pmb{\eta}_t\pmb{\eta}_t'=I_N$. Clearly, our S-BEKK model reduces to the standard
BEKK model in \citet{EK:1995} when $\bftau(x)$ is a constant matrix, and it includes the first-order S-BEKK model in \citet{HL:2010}
as a special case.

Let $\mathrm{tr}(A)$ and $\mathrm{det}(A)$ be the trace and determinant of a matrix $A$, respectively,
$\mathrm{vec}(A)$ be the vectorization of a matrix $A$ by stacking its columns, $A\otimes B$ be the Kronecker product between two matrices $A$ and $B$, and $A^{\otimes 2}=A\otimes A$. Denote $\pmb{\theta}=(\mathrm{vec}(A_1)',...,\mathrm{vec}(A_q)',\mathrm{vec}(B_1)',$ $...,\mathrm{vec}(B_p)')'\in\pmb{\Theta}$ be the unknown parameter of $\mathbf{u}_t$, and $\pmb{\Theta}\subset \mathbb{R}^{\mathrm{dim}(\ptheta)} $ be the parameter space, where
 $\mathrm{dim}(\ptheta)$ stands for the dimension of $\pmb{\theta}$.
Similar to the S-GARCH model, we consider the two-step estimation for the S-BEKK model. At step one,
we estimate $\bftau_{t}$ by $\widehat{\bftau}_{t}=\widehat{\bftau}(t/T)$, where
\begin{flalign*}
\widehat{\bftau}(x)=\frac{1}{T}{\sum_{s=1}^{T}K_h\Big(x-\frac{s}{T}\Big)\mathbf{y}_s\mathbf{y}_s'}.
\end{flalign*}
At step two, we consider the QMLE of $\pmb{\theta}_0$ given by
$
\widehat{\pmb{\theta}}_T=\arg\min_{\pmb{\theta}\in\pmb{\Theta}}\widehat{\mathbf{L}}_T(\pmb{\theta}),
$
where
$$
\widehat{\mathbf{L}}_T(\pmb{\theta})=\sum_{t=1}^{T}\widehat{\mathbf{l}}_t(\pmb{\theta})\quad\mbox{with}\quad \widehat{\mathbf{l}}_t(\pmb{\theta})=
\mathrm{tr}\big(\widehat{\mathbf{g}}_t(\pmb{\theta})^{-1}\widehat{\mathbf{u}}_t\widehat{\mathbf{u}}_t'\big)
+\log\det\big(\widehat{\mathbf{g}}_t(\pmb{\theta})\big).
$$
Here, $\widehat{\mathbf{g}}_t(\pmb{\theta})$ is calculated recursively by
$$
\widehat{\mathbf{g}}_t(\pmb{\theta})=I_N-\sum_{i=1}^{q}A_{i}A_{i}'-\sum_{j=1}^{p}B_{j}B_{j}'
+\sum_{i=1}^{q}A_{i}\widehat{\mathbf{u}}_{t-i}\widehat{\mathbf{u}}_{t-i}'A_{i}'+\sum_{j=1}^{p}B_{j}\widehat{\mathbf{g}}_{t-j}(\pmb{\theta})B_{j}'
$$
with $\widehat{\mathbf{u}}_t=\widehat{\bftau}_t^{-1/2}\mathbf{y}_t$, $t=1,...,T$, and some given constant initial values  $\widehat{\mathbf{u}}_0=\mathbf{u}_0,..., \widehat{\mathbf{u}}_{1-q}=\mathbf{u}_{1-q}$,
$\widehat{\mathbf{g}}_0(\pmb{\theta})=\mathbf{g}_0,..., \widehat{\mathbf{g}}_{1-p}(\pmb{\theta})=\mathbf{g}_{1-p}$.

To give the asymptotic distribution of $\widehat{\pmb{\theta}}_T$, we need the following notations.
Let $\mathcal{A}_i=A_i^{\otimes 2}$ for $i=1,...,q$, $\mathcal{B}_j=B_j^{\otimes 2}$ for $j=1,..., p$,  and $\mathbb{B}^k(1:N^2,1:N^2)$ be the upper-left $N^2\times N^2$ submatrix of $\mathbb{B}^k$, where
$$
\mathbb{B}=\left(\begin{matrix}
\mathcal{B}_1&\mathcal{B}_2&\cdots&\mathcal{B}_p\\
I_{N^2}&0&\cdots&0\\
0&I_{N^2}&\cdots&0\\
\vdots&\ddots&\ddots&\vdots\\
0&\cdots&I_{N^2}&0
\end{matrix}\right).
$$
Furthermore, let  $\pmb{\xi}_t=\mathrm{vec}(\pmb{\eta}_t\pmb{\eta}_t'-I_N)$, $\mathbf{\Upsilon}(x)=[\bftau(x)^{-1/4}\otimes \bftau(x)^{1/4}]$, $
\mathbf{\Omega}_0=I_{N^2}-\sum_{i=1}^{q}\mathcal{A}_{i0}-\sum_{j=1}^{p}\mathcal{B}_{j0}$, $\mathbf{\Gamma}_0=I_{N^2}-\sum_{j=1}^{p}\mathcal{B}_{j0}$, $\mathbf{Q}_t=(\mathbf{Q}_{t,1}',...,\mathbf{Q}_{t,\mathrm{dim}(\ptheta)}')'$, $\mathbf{N}=(\mathbf{N}_1',...,\mathbf{N}_{\mathrm{dim}(\ptheta)}')'$ and $\mathbf{M}=(\mathbf{M}_1',...,\mathbf{M}_{\mathrm{dim}(\ptheta)}')'$, where
\begin{flalign*}
\mathbf{Q}_{t,m}=&\mathrm{vec}\Big(\frac{\partial\mathbf{g}_t}{\partial\ptheta_m}\Big)'(\mathbf{g}_t^{-1/2})^{\otimes2},\quad\quad\quad \mathbf{N}_m=E\Big[\mathrm{vec}\Big(\frac{\partial\mathbf{g}_t}{\partial\ptheta_m}\Big)'(\mathbf{g}_t^{-1}\otimes I_N)\Big],\\
\mathbf{M}_m=&E\Big[\mathrm{vec}\Big(\frac{\partial {\mathbf{g}}_t}{\partial\ptheta_m}\Big)'(\mathbf{g}_t^{-1})^{\otimes2}\mathbf{T}_t\Big],\quad\mathbf{T}_t=\sum_{k=0}^{\infty}\mathbb{B}_0^k(1:N^2,1:N^2)\Big(\sum_{i=1}^{q}\mathcal{A}_{i0}
\Big)[I_N\otimes\mathbf{u}_{t-k-i}\mathbf{u}_{t-k-i}'].
\end{flalign*}
The following theorem establishes the asymptotic normality of $\widehat{\pmb{\theta}}_T$.

\begin{thm}\label{thm_BEKK}
	Suppose Assumptions C.1--C.7 in \citet{JLZ:2019} hold. Then,
	$$
	\sqrt{T}(\widehat{\ptheta}_T-\ptheta_0)\rightarrow_{\mathcal{L}} N(0,\mathbf{\Sigma})\mbox{ as }T\to\infty,
	$$
	where $\mathbf{\Sigma}=\big(E[\mathbf{Q}_t\mathbf{Q}_t']\big)^{-1}\big[\mathbf{J}_1+\mathbf{J}_2+\mathbf{J}_3+\mathbf{J}_3'\big]\big(E[\mathbf{Q}_t\mathbf{Q}_t']\big)^{-1}$ with
	\begin{flalign*}
	\mathbf{J}_1=&E\Big[\mathbf{Q}_t\mathrm{Var}(\pmb{\xi}_t)\mathbf{Q}_t'\Big],\\
	\mathbf{J}_2=&[\mathbf{M}-\mathbf{N}]\Big\{\int_{0}^{1}\mathbf{\Upsilon}(x)\mathbf{\Omega}_0^{-1}\mathbf{\Gamma}_0
E\Big[(\mathbf{g}_t^{1/2})^{\otimes2}\mathrm{Var}(\pmb{\xi}_t)(\mathbf{g}_t^{1/2})^{\otimes2}\Big]\mathbf{\Gamma}_0'\mathbf{\Omega}_0^{'-1}
\mathbf{\Upsilon}(x)dx\Big\}[\mathbf{M}-\mathbf{N}]',\\
		\mathbf{J}_3=&E\Big[\mathbf{Q}_t\mathrm{Var}(\pmb{\xi}_t)(\mathbf{g}_t^{1/2})^{\otimes2}\Big]
\Big\{\mathbf{\Gamma}_0'\mathbf{\Omega}_0^{'-1}\int_{0}^{1}\mathbf{\Upsilon}(x)dx\Big\}[\mathbf{M}-\mathbf{N}]'.
	\end{flalign*}
\end{thm}

When $p=q=1$, it can be shown that our asymptotic variance-covariance matrix $\mathbf{\Sigma}$ is equivalent to the one obtained in \citet{HL:2010}, but with a relatively simpler expression. Moreover, Theorem \ref{thm_BEKK} indicates that the effect of nonparametric part $\bftau_t$ on
 $\mathbf{\Sigma}$ is reflected by the term $\mathbf{\Upsilon}(x)$ existing in $\mathbf{J}_2$ and $\mathbf{J}_3$.
When $N=1$, we have $\mathbf{\Upsilon}(x)\equiv1$, and hence $\widehat{\ptheta}_T$ has the adaptiveness feature as demonstrated before.
When $N>1$, the form of $\bftau(x)$ has an impact on $\mathbf{\Upsilon}(x)$, except for some special cases
(e.g., $\bftau(x)=\tau(x) I_N$ with $\tau(x)>0$). Therefore, $\widehat{\ptheta}_T$ does not have the adaptiveness feature
 in the multivariate case.

\section{Simulations}
This section gives the simulation studies for the QMLE $\widehat{\theta}_{T}$ and the tests $LM_{T}$ and $Q_{T}(\ell)$. To facilitate it, we first
show how to choose the bandwidth $h$.

\subsection{Choice of bandwidth}\label{bandwidth}

The practical implementation of our entire methodologies needs to choose the bandwidth $h$. The methods in terms of
mean squared error criterion (see, e.g., \citet{HL:2010}) usually yield a bandwidth of order $T^{-1/5}$, which does not satisfy Assumption \ref{ass_bandwidth}. In what follows, we give a two-step cross-validation (CV) procedure to choose $h$ such that
Assumption \ref{ass_bandwidth} is satisfied.

\begin{alg}\label{alg1}
	{(CV bandwidth selection procedure)}
	\begin{enumerate}
		\item
		Set a pilot bandwidth $h_{0}=T^{-\lambda_0}$ with $\lambda_0\in(1/4,1/2)$, and then obtain the pilot estimates $\widehat{\tau}_{t,0}$ and $\widehat{u}_{t,0}$. Choose a pilot GARCH (or ARCH) model for the process $u_t$, and based on
		$\{\widehat{u}_{t,0}\}_{t=1}^{T}$,
		estimate this
		pilot model by the QMLE to get the pilot estimates $\{\widehat{g}_{t,0}\}_{t=1}^{T}$.

		\item
		With $\{\widehat{g}_{t,0}\}_{t=1}^{T}$,
		define a CV criterion as
		$$CV(h)=\sum_{t=1}^{T}\Big\{\frac{y_t^2}{{\widehat{\tau}_{-t}(h)\widehat{g}_{t,0}}}-1\Big\}^2,$$
		where $\widehat{\tau}_{-t}(h)$ is a leave-one-out estimate of $\tau_t$ with respect to the bandwidth $h$, based on all observations except for  $y_t$.
		Select our bandwidth as $h_{cv}=\arg\min_{h\in\mathcal{H}}CV(h)$, where
		$\mathcal{H}=[c_{\min}T^{-\lambda_0},c_{\max}T^{-\lambda_0}]$ with two positive constants
		$c_{\min}$ and $c_{\max}$.
	\end{enumerate}
\end{alg}

Let $\widehat{\mathrm{Var}}(y_t)$ be the sample variance of $\{y_{t}\}_{t=1}^{T}$.
To compute $h_{cv}$ in Algorithm \ref{alg1}, we suggest to choose $\lambda_0=2/7$, $c_{\min}=0.5\widehat{\mathrm{Var}}(y_t)^{\lambda_0}$ and $c_{\max}=3\widehat{\mathrm{Var}}(y_t)^{\lambda_0}$, which will be used and demonstrated with good performance in our simulation studies below.
For the pilot model in
Algorithm \ref{alg1}, it could be taken based on either some prior information or the Bayesian information criterion (BIC).

\subsection{Simulations for the estimation}\label{sec:simu_est}

In this subsection, we examine the finite-sample performance of the QMLE $\widehat{\theta}_{T}$. We generate
1000 replications of sample size $T=2000$ and $4000$ from the following two data generating processes (DGPs)
\begin{flalign*}
&\text{DGP 1 : \,\,The S-ARCH(2) model with }\alpha_{10}=\alpha_{20}=0.3; \\
&\text{DGP 2 : \,\,The S-GARCH(1, 1) model with }\alpha_{10}=0.1 \mbox{ and } \beta_{10}=0.8, 			
\end{flalign*}
where the function $\tau(x)$ is designed as follows
\begin{flalign}
\label{nochange}	\mbox{[No change]}&\quad \tau(x)=1;\\
\label{linearchange}	\mbox{[Linear change]}&\quad \tau(x)=1+2x;\\
\label{cyclicalchange}	\mbox{[Cyclical change]}&\quad \tau(x)=1+\sin(4\pi x)/2,
\end{flalign}
and the error $\eta_t$ follows $N(0,1)$, $st_{10}$, and $st_{5}$. Here, $st_{\nu}$ is the standardized Student-$t$ distribution with unit variance.

For each replication, we compute $\widehat{\theta}_{T}$ by using the Epanechnikov kernel $K(x)=\frac{3}{4}(1-x^2)\mathbf{1}(|x|\leq 1)$ and
choosing the bandwidth $h=h_{cv}$ according to Algorithm \ref{alg1} with the (G)ARCH model in DGP as the pilot model.
Table \ref{estimation} reports the sample bias, sample empirical standard deviation (ESD) and average asymptotic standard deviation (ASD) of $\widehat{\theta}_T$ based on 1000 replications for each DGP, where the ASD is calculated as in Remark \ref{rem_2}.
From Table \ref{estimation}, we find that (i) the biases of $\widehat{\theta}_{T}$ are small in each case; (ii)
regardless of the specification of $\tau(x)$ and the distribution of $\eta_t$, the values of ESD and ASD are close to each other, especially for large $T$; (iii) when the value of $T$ increases, the value of ESD decreases; (iv)
$\widehat{\theta}_{T}$ becomes less efficient with a larger value of ESD as the thickness of $\eta_t$ becomes heavier; (v) the value of ESD is almost invariant with respect to the specification of $\tau(x)$, meaning that $\widehat{\theta}_{T}$ is adaptive as expected.
Under the same settings as in Table \ref{estimation},  we also examine the finite-sample performance of the standard QMLE
in \citet{Bollerslev:1986}, and find that
when $\tau(x)\sim (\ref{nochange})$, the standard QMLE is more efficient than $\widehat{\theta}_{T}$;
but when $\tau(x)\sim (\ref{linearchange})$ or $(\ref{cyclicalchange})$, the standard QMLE suffers from larger bias
and discrepancy between ESD and ASD. For saving the space, these results for the standard QMLE are not reported here. Overall, our QMLE $\widehat{\theta}_{T}$ has a satisfactory performance in all considered cases, and the standard QMLE should not be used for the non-stationary S-GARCH model.

\begin{table}[!h]
	\centering
	\caption{The results $(\times100)$ of $\widehat{\theta}_{T}$ based on DGPs 1--2}	\label{estimation}
	\addtolength{\tabcolsep}{-2pt}  
	\begin{tabular}{cccccccccccccccccccc}
		\hline
		&                & \multicolumn{8}{c}{DGP 1: S-ARCH(2)}                                                                          &  & \multicolumn{8}{c}{DGP 2: S-GARCH(1,1)}                                                                      \\ \cline{3-10} \cline{12-19}
		&               & \multicolumn{2}{c}{$N(0,1)$} &  & \multicolumn{2}{c}{$st_{10}$} &  & \multicolumn{2}{c}{$st_{5}$}  &  & \multicolumn{2}{c}{$N(0,1)$} &  & \multicolumn{2}{c}{$st_{10}$} &  & \multicolumn{2}{c}{$st_{5}$} \\ \cline{3-4} \cline{6-7} \cline{9-10} \cline{12-13} \cline{15-16} \cline{18-19}
		$T$              &      & $\alpha_{10}$      & $\alpha_{20}$     &  & $\alpha_{10}$ & $\alpha_{20}$ &  & $\alpha_{10}$ & $\alpha_{20}$ &  & $\alpha_{10}$      & $\beta_{10}$      &  & $\alpha_{10}$  & $\beta_{10}$ &  & $\alpha_{10}$ & $\beta_{10}$ \\ \hline\\
		&          &      &                    &                   &  &               &                 & \multicolumn{5}{c}{Panel A: $\tau(x)\sim (\ref{nochange})$}                                 &  &                &              &  &               &              \\
		2000   & Bias & -0.63              & -0.72             &  & -1.13         & -1.21         &  & -1.44         & -2.26         &  & -0.11              & -3.45             &  & -0.18          & -3.49        &  & 0.05          & -4.31        \\
		& ESD  & 3.90               & 3.93              &  & 4.73          & 4.79          &  & 7.17          & 7.44          &  & 2.02               & 6.35              &  & 2.30           & 7.13         &  & 3.04          & 8.39         \\
		& ASD  & 3.96               & 3.96              &  & 4.96          & 4.95          &  & 7.60          & 7.44          &  & 2.10               & 5.62              &  & 2.38           & 6.25         &  & 3.19          & 7.85         \\
		4000  & Bias & -0.36              & -0.45             &  & -0.59         & -0.56         &  & -1.35              & -1.58              &  & -0.08              & -1.82             &  & -0.09          & -1.86        &  & -0.14         & -1.92        \\
		&           ESD  & 2.81               & 2.78              &  & 3.37          & 3.38          &  &  5.76             & 5.98              &  & 1.38               & 3.76              &  & 1.56           & 4.04         &  & 2.05          & 4.80         \\
		&           ASD  & 2.85               & 2.64              &  & 3.67          & 3.68          &  &5.69               & 5.69              &  & 1.45               & 3.53              &  & 1.66           & 3.90         &  & 2.22          & 4.89         \\
		&                &                    &                   &  &               &               &  &               &               &  &                    &                   &  &                &              &  &               &              \\
		&          &      &                    &                   &  &               &                 & \multicolumn{5}{c}{Panel B: $\tau(x)\sim (\ref{linearchange})$}                                 &  &                &              &  &               &              \\
		2000   & Bias & -0.30              & -0.36             &  & -0.75         & -0.83         &  &-1.11               & -1.93              &  & 0.12               & -1.99             &  & 0.06           & -2.17        &  & 0.31          & -3.41        \\
		&           ESD  & 3.90               & 3.98              &  & 4.71          & 4.74          &  &  7.12             &   7.34            &  & 2.02               & 5.96              &  & 2.29           & 6.17         &  & 3.19          & 7.99         \\
		&           ASD  & 3.82               & 3.98              &  & 5.01          & 5.00          &  &     7.72          &  7.55              &  & 2.05               & 4.96              &  & 2.34           & 5.57         &  & 3.23          & 7.32         \\
		4000  & Bias & -0.04              & -0.13             &  & -0.28         & -0.23         &  & -1.05              &  -1.25             &  & 0.11               & -0.82             &  & 0.08           & -1.09        &  & 0.05          & -1.36        \\
		&           ESD  & 2.80               & 2.76              &  & 3.31          & 3.36          &  &   5.69            &   5.91            &  & 1.40               & 3.24              &  & 1.57           & 3.64         &  & 2.20          & 4.65         \\
		&           ASD  & 2.85               & 2.85              &  & 3.71          & 3.72          &  &  5.79             &   5.79            &  & 1.42               & 3.21              &  & 1.64           & 3.62         &  & 2.28          & 4.66         \\
		&                &                    &                   &  &               &               &  &               &               &  &                    &                   &  &                &              &  &               &              \\
		&          &      &                    &                   &  &               &                 & \multicolumn{5}{c}{Panel C: $\tau(x)\sim (\ref{cyclicalchange})$}                                 &  &                &              &  &               &              \\
		2000   & Bias & 0.04               & -0.07             &  & -0.37         & -0.49         &  & -0.68              &   -1.47            &  & 0.18               & -1.46             &  & 0.13           & -1.79        &  & 0.36          & -2.53        \\
		&           ESD  & 3.92               & 3.91              &  & 4.71          & 4.70          &  &   6.97            &   7.16            &  & 2.07               & 5.48              &  & 2.32           & 6.05         &  & 2.49          & 5.91         \\
		&           ASD  & 3.97               & 3.97              &  & 4.91          & 4.98          &  &   7.61            & 7.47              &  & 2.03               & 4.78              &  & 2.31           & 5.40         &  & 2.53          & 5.54         \\
		4000   & Bias & 0.27               & 0.19              &  & 0.07          & 0.08          &  & -0.58         & -0.85        &  & 0.19               & -0.23             &  & 0.18           & -0.62        &  & 0.16          & -0.86        \\
		&           ESD  & 2.69               & 2.81              &  & 3.39          & 3.35          &  & 5.65          & 5.82         &  & 1.41               & 3.19              &  & 1.59           & 3.55         &  & 2.28          & 4.68         \\
		&           ASD  & 2.78               & 2.84              &  & 3.66          & 3.67          &  & 5.71          & 5.70          &  & 1.40               & 3.05              &  & 1.62           & 3.47         &  & 2.25          & 4.41         \\
		\hline
	\end{tabular}
\end{table}

\subsection{Simulations for the testing} \label{subtest}

In this subsection, we examine the finite-sample performance of $LM_{T}$ and $Q_{T}(\ell)$. We generate
1000 replications of sample size $T=2000$ and $4000$ from the following two DGPs
\begin{flalign*}
&\text{DGP 3 : \,\,The S-GARCH(1, 2) model with }\alpha_{10}=\beta_{10}=0.3\mbox{ and }\alpha_{20}=0.03k;\\
&\text{DGP 4 : \,\,The S-GARCH(2, 1) model with }\alpha_{10}=\beta_{10}=0.3\mbox{ and }\beta_{20}=0.03k,
\end{flalign*}
where $k=0,1,..., 10$, $\tau(x)$ is designed as in DGPs 1--2, and $\eta_t\sim N(0, 1)$.
For each DGP, the model with respect to $k=0$ is taken as its null model. That is, the S-GARCH($1, 1$) model is the null model for both DGP 3 and DGP 4.

Next, we fit each replication by its related null model, and then apply $LM_{T}$ to detect
the null hypothesis of $k=0$ as well as
$Q_{T}(\ell)$ to check whether this fitted null model is adequate.
Based on 1000 replications, the empirical power of
$LM_{T}$ and $Q_{T}(\ell)$ is plotted in Fig\,\ref{ARCHtest} and Fig\,\ref{GARCHtest} for
DGP  3 and DGP 4, respectively,
where we take the level $\alpha=5\%$ and the lag $\ell=6, 9$, and $12$, and
the sizes of both tests correspond to the results for $k=0$.

\begin{figure}[!h]
	\centering
	\includegraphics[width=15cm,height=8cm]{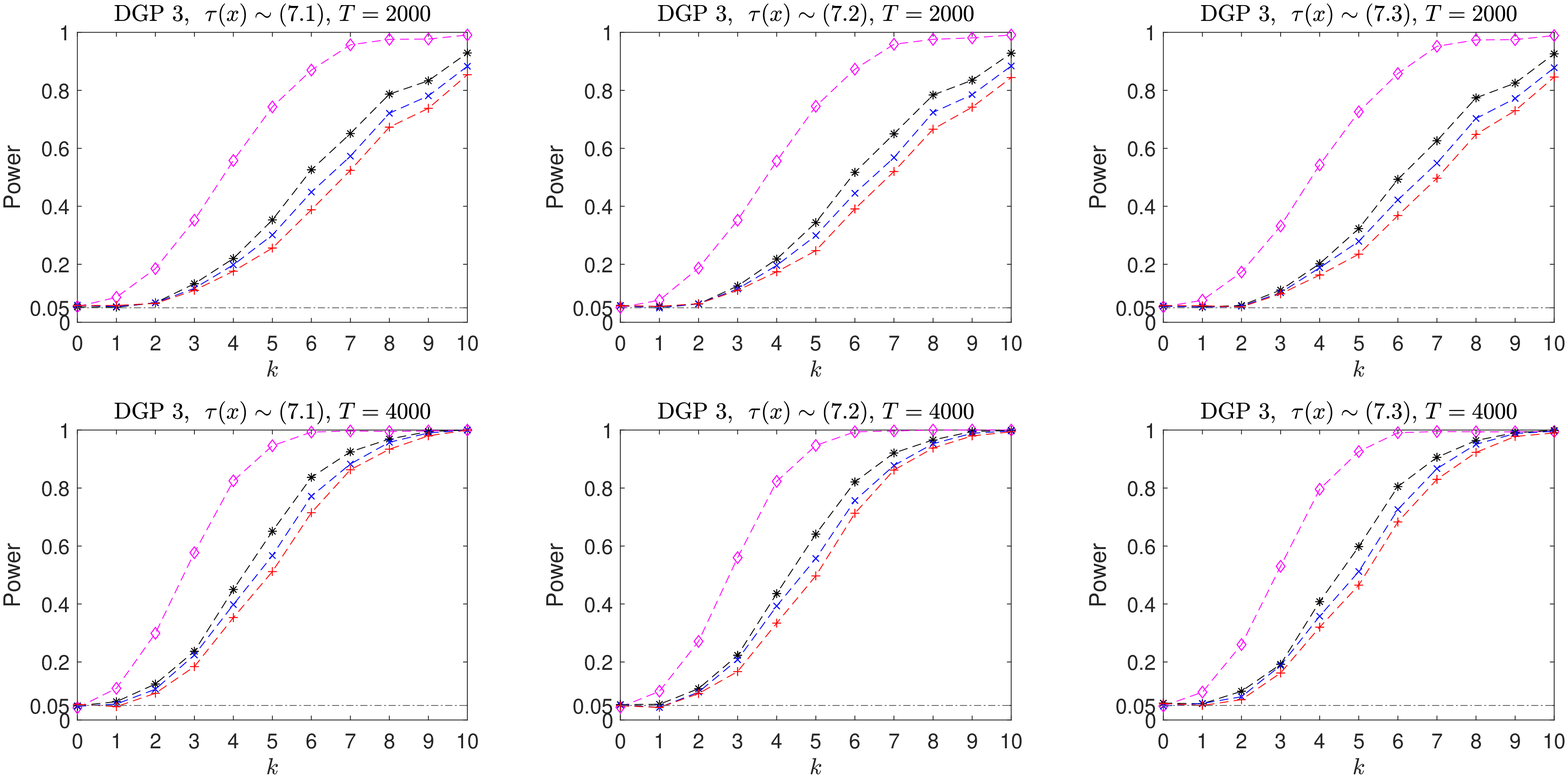}
	\caption{Power across $k$ in DGP 3 for $LM_{T}$ (diamond ``$\diamond$'' marker) and $Q_{T}(\ell)$ with $\ell=6$ (star ``$\ast$'' marker), $\ell=9$ (cross ``$\times$'' marker), and $\ell=12$ (plus ``$+$'' marker).	The horizontal dash-dotted line corresponds to the level $5\%$. Upper Panel: $T=2000$; bottom Panel:  $T=4000$.}
	\label{ARCHtest}
\end{figure}

\begin{figure}[!h]
	\centering
	\includegraphics[width=15cm,height=8cm]{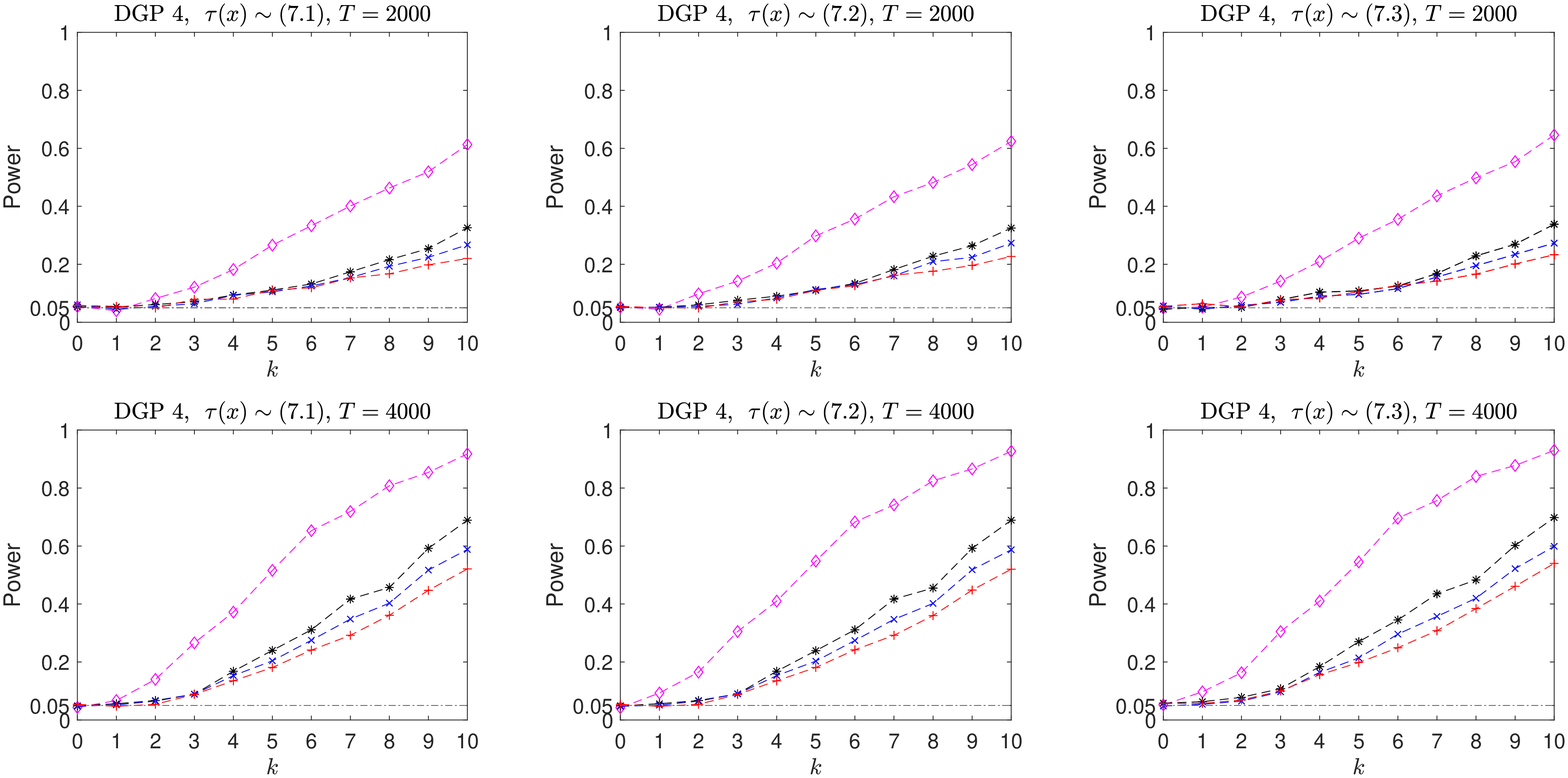}
	\caption{Power across $k$ in DGP 4.
		The descriptions are as for Fig~\ref{ARCHtest}.
}
	\label{GARCHtest}
\end{figure}

From Figs\,\ref{ARCHtest}--\ref{GARCHtest}, we can find that (i) all tests have precise sizes;
(ii) the power of all tests becomes large as the value of $T$ or $k$ increases; (iii) $LM_{T}$ is more powerful than all $Q_{\ell}$, and $Q_6$ is generally more powerful than $Q_9$ and $Q_{12}$;
(iv) all tests are more powerful to detect the mis-specification of ARCH part in DGP  3 than
the mis-specification of GARCH part in DGP 4; (v) all tests are adaptive, since their
power is unaffected by the form of $\tau(x)$.
In summary, all tests have a good performance especially for large $T$.

\subsection{Comparison with three-step estimation method}
In this subsection, we compare the finite-sample performance of $\widehat{\theta}_T$ and the three-step estimator $\widecheck{\theta}_{T}$ by investigating their bias difference and efficiency ratio (componentwisely) defined respectively as
$$
d(\gamma)=(|\mbox{the Bias of }\widehat{\gamma}_T|-|\mbox{the Bias of }\widecheck{\gamma}_T|)\times 100
\mbox{ and }
R(\gamma)=\frac{\mbox{the ESD of }\widehat{\gamma}_{T}}{\mbox{the ESD of }\widecheck{\gamma}_{T}},
$$
where $\gamma$ denotes any entry of $\theta_0$, and the Bias and ESD of each estimator are computed based on 1000 replications.
We calculate the values of  $d(\gamma)$ and $R(\gamma)$ under the same simulation settings as in Subsection \ref{sec:simu_est},
and  only report the results for the case of
$\tau(x)\sim$ (\ref{nochange}) in Table \ref{comparehf} due to the adaptiveness of $\widehat{\theta}_T$ and $\widecheck{\theta}_{T}$.
From Table \ref{comparehf}, we can find that (i) both estimators have a comparable bias performance; (ii)
when $\eta_t\sim N(0, 1)$, $\widecheck{\theta}_T$ is more (or less) efficient than $\widehat{\theta}_T$ in DGP 1 (or DGP 2),
indicating that $\widecheck{\theta}_T$ does not achieve the semiparametric efficiency bound as indicated in Theorem \ref{thm_improve};
(iii) when $\eta_t$ has a heavier distribution (e.g., $\eta_t \sim st_{5}$), $\widehat{\theta}_T$ exhibits more efficiency advantage over
$\widecheck{\theta}_{T}$. In summary, our simulation results suggest that it is unnecessary to
further update $\widehat{\theta}_T$ to $\widecheck{\theta}_T$.

\begin{table}[!h]
	\centering
	\caption{The results of $d(\gamma)$ and $R(\gamma)$ based on DGPs 1--2 with $\tau(x)\sim$ (\ref{nochange})}	\label{comparehf}
	\addtolength{\tabcolsep}{-2pt}  
	\begin{tabular}{ccccccccccccccccccc}
		\hline
		&                  & \multicolumn{8}{c}{DGP 1: S-ARCH(2)}                                                                         &  & \multicolumn{8}{c}{DGP 2: S-GARCH(1,1)}                                                                     \\ \cline{3-10} \cline{12-19}
		&                  & \multicolumn{2}{c}{$N(0,1)$} &  & \multicolumn{2}{c}{$st_{10}$} &  & \multicolumn{2}{c}{$st_5$}    &  & \multicolumn{2}{c}{$N(0,1)$} &  & \multicolumn{2}{c}{$st_{10}$} &  & \multicolumn{2}{c}{$st_5$}   \\ \cline{3-4} \cline{6-7} \cline{9-10} \cline{12-13} \cline{15-16} \cline{18-19}
		$T$    &                  & $\alpha_{10}$      & $\alpha_{20}$     &  & $\alpha_{10}$ & $\alpha_{20}$ &  & $\alpha_{10}$ & $\alpha_{20}$ &  & $\alpha_{10}$      & $\beta_{10}$      &  & $\alpha_{10}$  & $\beta_{10}$ &  & $\alpha_{10}$ & $\beta_{10}$ \\ \hline
		&                  &                    &                   &  &               &               &  &               &               &  &                    &                   &  &                &              &  &               &              \\
		2000 & $d(\gamma)$ & 0.08              & 0.08             &  & -0.23        & -0.22        &  & -0.49       & -0.53        &  & -0.06             & -0.43            &  & -0.05         & -0.47       &  & 0.04         & -0.34       \\
		& $R(\gamma)$ & 1.01              & 1.02             &  & 1.01         & 0.99         &  & 0.46         & 0.48         &  & 0.99              & 0.94             &  & 0.99          & 0.95        &  & 0.89         & 0.91        \\
		4000 & $d(\gamma)$ & 0.05              & 0.05             &  & -0.13        & -0.15        &  & -0.68        & -1.15        &  & -0.02             & -0.19            &  & -0.02         & -0.03       &  & 0.02         & -0.31       \\
		& $R(\gamma)$ & 1.04              & 1.04             &  & 1.03         & 1.03         &  & 0.40         & 0.24         &  & 0.99              & 0.96             &  & 0.99          & 0.96        &  & 0.97         & 0.96        \\
\hline
	\end{tabular}
\end{table}

\subsection{Comparison with the VT method}

In this subsection, we compare the finite-sample performance of $\widehat{\theta}_{T}$, $LM_T$ and  $Q_T(\ell)$
with those of $\overline{\theta}_{T}$, $LM_{T}^{vt}$ and $Q_T^{vt}(\ell)$, respectively, where
$\overline{\theta}_{T}$ defined in (\ref{vt_qmle}) is the QMLE from the VT method, and
$LM_{T}^{vt}$ and $Q_T^{vt}(\ell)$ are defined in the same way as $LM_T$ and  $Q_T(\ell)$ with
$\widehat{\theta}_{T}$ replaced by $\overline{\theta}_{T}$. Note that when the S-GARCH($p, q$) model is stationary,
$\overline{\theta}_{T}$ is asymptotically normal, and $LM_{T}^{vt}$ and $Q_T^{vt}(\ell)$ have the same limiting null distributions
as those of $LM_T$ and  $Q_T(\ell)$.

First, we compare the efficiency of $\widehat{\theta}_{T}$ and $\overline{\theta}_{T}$ by looking at the following ratio
$$R_{qmle}(\gamma)=\frac{\mbox{the ESD of }\widehat{\gamma}_{T}}{\mbox{the ESD of }\overline{\gamma}_{T}},$$
where the ESD of each estimator is computed based on 1000 replications.
Table \ref{estimationcompare} reports the values of $R_{qmle}(\gamma)$ when the DGP
is a stationary S-ARCH(2) (or S-GARCH($1, 1$)) model with $\tau(x)\sim (\ref{nochange})$, $\eta_t\sim N(0, 1)$, and three different choices of
$\theta_0$.
From this table, we find that as expected, all the values of $R_{qmle}(\gamma)$ are close to 1, indicating that $\widehat{\theta}_{T}$ and $\overline{\theta}_{T}$ have the same asymptotic efficiency when the S-GARCH model is stationary.

\begin{table}[!h]
	\centering
	\addtolength{\tabcolsep}{-0.2pt}
	\caption{The value of $R_{qmle}(\gamma)$ when the S-GARCH model is stationary}\label{estimationcompare}
	\begin{tabular}{ccccccccc}
		\hline
		& \multicolumn{8}{c}{\multirow{2}{*}{DGP: S-ARCH(2) with $\tau(x)\sim (\ref{nochange})$ and $\eta_t\sim N(0, 1)$}}                                                                       \\
		\multicolumn{8}{c}{}   \\
		\cline{2-9}
		&  \multicolumn{2}{c}{\multirow{2}{*}{$(\alpha_{10}, \alpha_{20})=(0.5, 0.1)$}}   &&
		\multicolumn{2}{c}{\multirow{2}{*}{$(\alpha_{10}, \alpha_{20})=(0.4, 0.2)$}}  &&
		\multicolumn{2}{c}{\multirow{2}{*}{$(\alpha_{10}, \alpha_{20})=(0.3, 0.3)$}} \\
		\multicolumn{8}{c}{}   \\				
		\cline{2-3} \cline{5-6} \cline{8-9}				
		\multicolumn{1}{c}{\multirow{2}{*}{$T$}}    & \multicolumn{1}{c}{\multirow{2}{*}{$R_{qmle}(\alpha_{10})$}}& \multicolumn{1}{c}{\multirow{2}{*}{$R_{qmle}(\alpha_{20})$}}
		&& \multicolumn{1}{c}{\multirow{2}{*}{$R_{qmle}(\alpha_{10})$}} & \multicolumn{1}{c}{\multirow{2}{*}{$R_{qmle}(\alpha_{20})$}}
		&& \multicolumn{1}{c}{\multirow{2}{*}{$R_{qmle}(\alpha_{10})$}} & \multicolumn{1}{c}{\multirow{2}{*}{$R_{qmle}(\alpha_{20})$}} \\
		\multicolumn{8}{c}{}   \\
		\cline{1-9}\\
		2000 & 0.951         & 1.036         && 1.043        & 0.978         && 1.005        & 0.997        \\
		4000 & 0.963         & 1.029         && 1.007        & 1.002         && 1.013        & 1.005        \\\\				
		& \multicolumn{8}{c}{\multirow{2}{*}{DGP: S-GARCH($1, 1$) with $\tau(x)\sim (\ref{nochange})$ and $\eta_t\sim N(0, 1)$}}                                                                       \\
		\multicolumn{8}{c}{}   \\
		\cline{2-9}
		&  \multicolumn{2}{c}{\multirow{2}{*}{$(\alpha_{10}, \beta_{10})=(0.1, 0.8)$}}   &&
		\multicolumn{2}{c}{\multirow{2}{*}{$(\alpha_{10}, \beta_{10})=(0.2, 0.7)$}}  &&
		\multicolumn{2}{c}{\multirow{2}{*}{$(\alpha_{10}, \beta_{10})=(0.3, 0.5)$}} \\
		\multicolumn{8}{c}{}   \\				
		\cline{2-3} \cline{5-6} \cline{8-9}				
		\multicolumn{1}{c}{\multirow{2}{*}{$T$}}    & \multicolumn{1}{c}{\multirow{2}{*}{$R_{qmle}(\alpha_{10})$}}& \multicolumn{1}{c}{\multirow{2}{*}{$R_{qmle}(\beta_{10})$}}
		&& \multicolumn{1}{c}{\multirow{2}{*}{$R_{qmle}(\alpha_{10})$}} & \multicolumn{1}{c}{\multirow{2}{*}{$R_{qmle}(\beta_{10})$}}
		&& \multicolumn{1}{c}{\multirow{2}{*}{$R_{qmle}(\alpha_{10})$}} & \multicolumn{1}{c}{\multirow{2}{*}{$R_{qmle}(\beta_{10})$}} \\
		\multicolumn{8}{c}{}   \\
		\cline{1-9}\\
		2000 &  1.003        & 1.022       && 1.036        & 1.104        && 0.995        & 1.049        \\
		4000 & 0.980         & 1.089        && 1.012         & 1.054        && 0.971        & 1.018        \\ \hline
	\end{tabular}
\end{table}

Second, we compare the power of $LM_{T}$ and $LM_{T}^{vt}$ and that of
$Q_{T}(\ell)$ and $Q_{T}^{vt}(\ell)$ by looking at the following two ratios
$$R_{lm}=\frac{\mbox{the power of }LM_{T}}{\mbox{the power of }LM_{T}^{vt}}\quad\mbox{ and }\quad
R_{q}(\ell)=\frac{\mbox{the power of }Q_{T}(\ell)}{\mbox{the power of }Q_{T}^{vt}(\ell)},$$
where the power of each test is computed based on 1000 replications.
Table \ref{VTpower} reports the values of $R_{lm}$ and $R_{q}(\ell)$ (for $\ell=6, 9$, and $12$), when
the data are generated from a stationary S-GARCH($1, 2$) model in DGP 3 with
$\tau(x)\sim (\ref{nochange})$. The results for DGP 4 are quite similar and hence omitted to save space.
From Table \ref{VTpower}, we can see that (i) the values of $R_{lm}$ are close to 1 in all examined cases;
(ii) when the value of $T$ or $k$ is small, the values of $R_{q}(\ell)$ are slightly less than one, meaning that
$Q_T^{vt}(\ell)$ could be more powerful than $Q_T(\ell)$; (ii) when the value of $T$ or $k$ becomes large,
the power advantage of $Q_T^{vt}(\ell)$ disappears as the values of $R_{q}(\ell)$ are close to 1.
These findings demonstrate that when the S-GARCH model is stationary,
our two tests have the same power performance as their counterparts from the VT method especially for large $T$. We also highlight that
when the S-GARCH model is non-stationary, our unreported results show that  $LM^{vt}_{T}$ and $Q_T^{vt}(\ell)$
can cause a severe over-sized problem, and hence they can not be used in this case.

\begin{table}[!h]
	\addtolength{\tabcolsep}{-2.8pt}
	\caption{The values $R_{lm}$ and $R_{q}(\ell)$ based on a stationary S-GARCH($1, 2$) model in DGP 3}\label{VTpower}
	\begin{tabular}{lcccccccccccc}
		\hline
		&\diagbox[width=5em,trim=l]{$T$}{$k$}    & 0 & 1      & 2      & 3      & 4       & 5      & 6      & 7      & 8      & 9      & 10     \\
		\hline
		&\multicolumn{12}{c}{}                                                                               \\
		$R_{lm}$
		&2000 &    1.019&	1.162&	1.081&	1.029&	1.018&	0.997&	0.996&	0.989&	0.98&	0.981&	0.991
		\\
		&4000 & 0.981	&1.123	&1.027&	0.995	&0.981	&0.989	&1.007&	1.001&	1.000&	1.000&	1.000\\
		&\multicolumn{12}{c}{}                                                                               \\
		$R_{q}(6)$
		&2000 &  1.232&	1.019&	0.761&	0.852&	0.916&	0.956&	1.073	&0.987	&0.987&	0.959&	0.988
		\\
		 &4000 &    1.184	&0.889	&0.799	&0.810	&0.907	&0.919	&0.962&	0.968	&0.991&	0.995	&0.999
		\\
		&\multicolumn{12}{c}{}                                                                               \\
		$R_{q}(9)$
		&2000 &  1.282&	0.843&	0.798	&0.943	&0.990	&0.961	&1.0346	&0.979&	0.993&	0.951&	0.983
		\\&4000 & 1.021&1.036&	0.754&	0.839&	0.899&	0.903&	0.940&	0.971&	0.985&	0.9997&	0.996		
		\\
		&\multicolumn{12}{c}{}                                                                               \\
		$R_{q}(12)$
		&2000 &  1.160
		& 0.698 & 0.833 & 0.916 & 1.053 & 0.870 & 1.021 & 0.990 & 0.979 & 0.963 & 0.982 \\&4000 &1.056&	1.087&	0.765&	0.880&	0.906	&0.908&	0.955&	 0.976&	0.950&	0.991&	0.999
		\\ \hline
	\end{tabular}
\end{table}

	\section{Applications}
	In this section, we re-study the US dollar to Indian rupee (USD/INR)  exchange rate series and FTSE-index series in \citet{Truquet:2017},
	with respect to in-sample fitting and out-of-sample prediction.

	\subsection{USD/INR exchange rate}
	
	This subsection considers the USD/INR exchange rate series from December 19th, 2005 to February 18th, 2015.
	The log returns (in percentage) of this series having $T=2301$ observations in total are denoted by $\{y_t\}$, and they are plotted in  the upper panel of
	Fig\,\ref{rupee}.
	 We  apply the non-parametric strict stationarity test in \cite{Hong:2017} (with the same settings as in their simulation) to $\{y_t\}$ and find this test statistic has a p-value close to zero, indicating a strong evidence against the strict stationarity. Thus, using a non-stationary model to fit this series seems appropriate.
	 In \citet{Truquet:2017}, this return series is fitted by a semiparametric
	ARCH(1) model with a time-varying intercept and a constant lag-1 ARCH parameter.
	Motivated by this, we use an ARCH(1) model as the pilot model in Algorithm \ref{alg1} to choose the bandwidth $h=0.0358$, and then calculate the series $\{\widehat{u}_{t}\}$.  Based on $\{\widehat{u}_{t}\}$, the BIC selects $p=q=1$ for the S-GARCH model, and hence we fit this return series by
	the S-GARCH($1, 1$) model with $\widehat{\alpha}_{1T}=0.0762_{(0.0231)}$, $\widehat{\beta}_{1T}=0.8443_{(0.0475)}$, and
	$\widehat{\tau}_t$ being plotted in the middle panel of Fig\,\ref{rupee}, where the values in parentheses are
	the asymptotic standard errors, and the bandwidth $h=0.0833$ is re-chosen by using a GARCH($1, 1$) model as
	the pilot model in Algorithm \ref{alg1}. For this fitted S-GARCH($1, 1$) model, the p-values of
	the portmanteau tests $Q_{T}(6)$, $Q_{T}(9)$, and $Q_{T}(12)$ are 0.6472, 0.7530, and 0.8268, respectively, implying that
	our fitted short run GARCH($1, 1$) component is adequate.
In view of the plot of $\{\widehat{\tau}_t\}$ in Fig\,\ref{rupee}, we can find that the long run component $\tau_t$
has relatively larger values around years 2009 and 2014.  Moreover, we also plot the estimated volatilities
based on either S-GARCH or GARCH model in the bottom panel of Fig\,\ref{rupee}, from which we can see that compared with the S-GARCH model, the GARCH model tends to underestimate the volatilities during 2008-2009 and 2013-2014, and overestimate the volatilities during other periods.
\begin{figure}[!h]
		\centering
		\includegraphics[width=12cm,height=10cm]{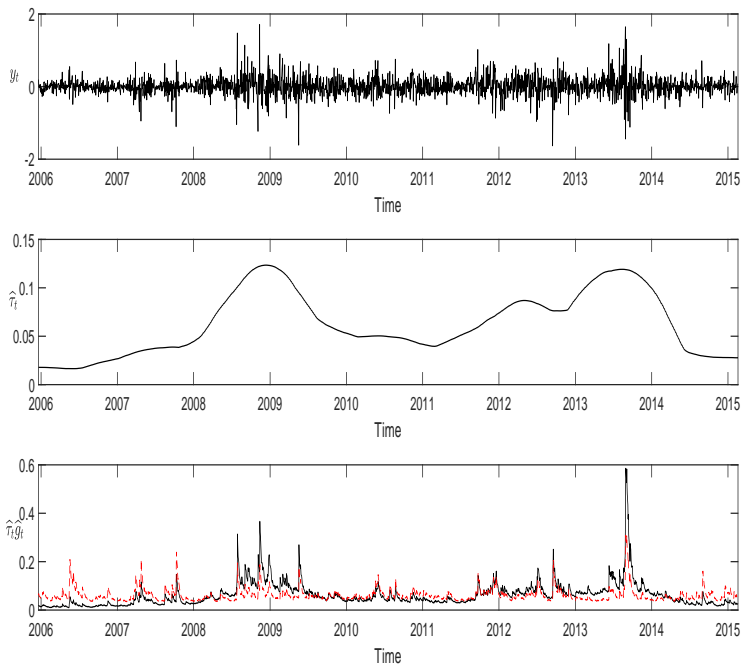}
		\caption{The plot of log returns $\{y_t\}$ (upper panel), estimated long-run components $\{\widehat{\tau}_t\}$ (middle panel), and estimated volatilities $\{\widehat{\tau}_t\widehat{g}_t\}$ (bottom panel) based on S-GARCH model (solid line) and GARCH model (dotted line) for USD/INR series. Here,  $\widehat{\tau}_t$ is computed by using the Epanechnikov kernel with $h=0.0833$ for the S-GARCH model, and
$\widehat{\tau}_t\equiv \overline{\tau}_T$ and $\widehat{g}_t=\overline{g}_{t}(\overline{\theta}_T)$ in (\ref{vt_qmle}) for the GARCH model. }
		\label{rupee}
	\end{figure}

	\subsection{FTSE-index}
	
	This subsection considers the FTSE-index series from January 4th, 2005 to March
	4th, 2015. We study the log returns of this index series with $T=2568$ observations in total, which is denoted by
	$\{y_t\}$ and plotted in the upper panel of Fig\,\ref{FTSE}.
	As the previous example, we use the non-parametric strict stationarity test in \cite{Hong:2017} to $\{y_t\}$, and find a strong evidence (with the p-value close to zero) against the strict stationarity.  Since \citet{Truquet:2017}
	suggested a semiparametric
	ARCH(5) model with a time-varying intercept and constant ARCH parameters to fit this return series,
	we take an ARCH(5) model as the pilot model in Algorithm \ref{alg1}, and then select
	the bandwidth $h=0.0865$ as a result. Based on this choice of $h$, we compute $\{\widehat{u}_{t}\}$ and select $p=q=1$ according to the BIC.
	Hence, we fit this return series by the S-GARCH($1, 1$) model with $\widehat{\alpha}_{1T}=0.1098_{(0.0165)}$, $\widehat{\beta}_{1T}=0.8433_{(0.0233)}$, and
	$\widehat{\tau}_t$ being plotted in the middle panel of Fig\,\ref{FTSE}, where
	the bandwidth $h=0.0941$ is re-chosen by using a S-GARCH($1, 1$) model as
	the pilot model in Algorithm \ref{alg1}. Further, the portmanteau tests $Q_{T}(6)$, $Q_{T}(9)$, and $Q_{T}(12)$
	(with p-values equal to $0.5326$, $0.5335$, and $0.2800$, respectively) suggest that
	this fitted short run GARCH($1, 1$) component is adequate. From the middle panel of Fig\,\ref{FTSE}, we find that the long run component $\tau_t$ for the FTSE return series
	only has a clear peak around 2009. This may imply that the stock market index series
	has a different long run structure with the exchange rate series. Moreover, we also plot the estimated volatilities based on either  S-GARCH or  GARCH model in the bottom panel of Fig\,\ref{FTSE}, from which we can see that the estimated volatilities from two models are quite close except  around years 2008-2009, during which the GARCH model tends to underestimate the volatilities.
\begin{figure}[!h]
		\centering
		\includegraphics[width=15cm,height=10cm]{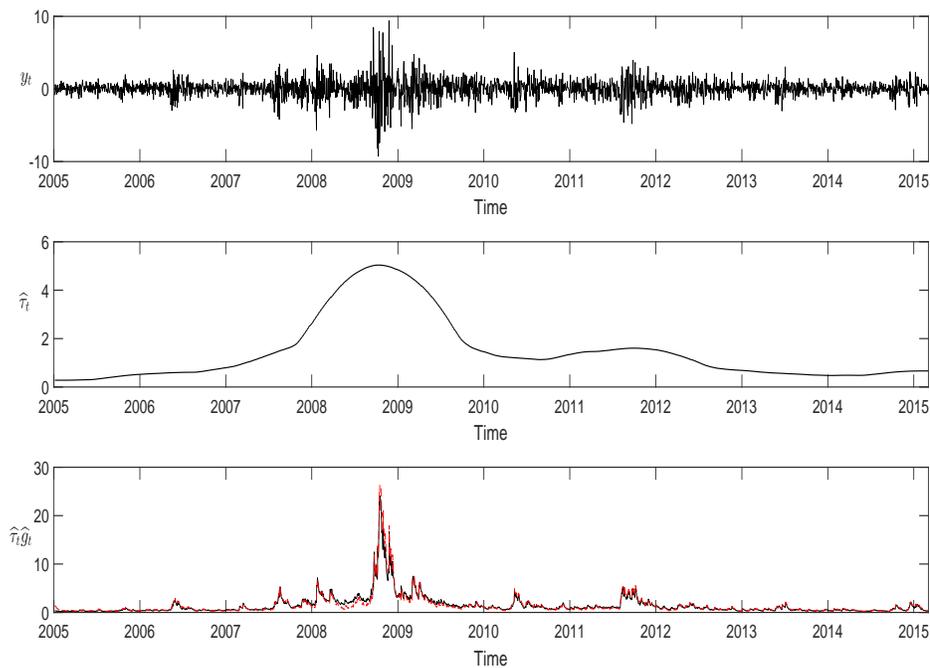}
		\caption{The plots for FTSE series with the same descriptions as in Fig\,\ref{rupee}. Here,  $\widehat{\tau}_t$ is computed by using the Epanechnikov kernel with $h=0.0941$ for the S-GARCH model.}
		\label{FTSE}
	\end{figure}

	\subsection{Forecasting comparisons}
	
	This subsection makes a forecasting comparison among S-GARCH($1, 1$) model,
	S-ARCH($q$) model, GARCH($1, 1$) model in \citet{Bollerslev:1986}, and LS-ARCH($q$) model (i.e., the locally stationary ARCH($q$) model) in  \citet{FSR:2008} for the USD/INR and FTSE return series.
	Note that the S-ARCH($q$) model can locally approximate the semiparametric ARCH($q$) model in \citet{Truquet:2017}, where
	$q=1$ (or 5) is suggested for the USD/INR (or FTSE) return series. Hence, we follow \citet{Truquet:2017} to select $q$ for the S-ARCH($q$) and LS-ARCH($q$) models.
	
	Next, we compare all four models in terms of the averaged QLIKE loss function in \citet{Patton:2011}.
	Specifically, we
	use the in-sample data  $\{y_t\}_{t=1}^{T_0}$ to make a $t_0$-step ahead forecast
	$\widehat{y}_{{T_0}+t_0|T_0}^2$ for the out-of-sample data point $y_{T_0+t_0}^2$, and then compute the averaged QLIKE by
	 $$\mbox{QLIKE}(t_0)=\frac{1}{T-t_0-1499}\sum_{T_0=1500}^{T-t_0}\log\widehat{y}_{{T_0}+t_0|T_0}^2+\frac{y_{T_0+t_0}^2}{\widehat{y}_{{T_0}+t_0|T_0}^2}.$$
	The model with the smaller value of $\mbox{QLIKE}(t_0)$ has the better $t_0$-step ahead forecasting performance.
	
	Moreover, we introduce how each model computes $\widehat{y}_{{T_0}+t_0|T_0}^2$. For the S-GARCH($1, 1$) model,
	we fit the model via the two-step estimation based on the in-sample data $\{y_t\}_{t=1}^{T_0}$, where the bandwidth $h$ is chosen by
	Algorithm  \ref{alg1} with a pilot GARCH(1, 1) model. With the kernel estimate
	$\widehat{\tau}_{T_0}$ and QMLE $\widehat{\theta}_{T_0}$, we then obtain  $\widehat{y}_{T_0+t_0|T_0}^2=\widehat{\tau}_{T_0}g_{T_0+t_0|T_0}(\widehat{\theta}_{T_0})$, where
	$g_{T_0+t_0|T_0}(\widehat{\theta}_{T_0})$ computed as for volatility
	prediction in the GARCH($1, 1$) model is the $t_0$-step ahead prediction of $g_{T_0+t_0}$. A similar way is used for the S-ARCH($q$) model
	to compute $\widehat{y}_{{T_0}+t_0|T_0}^2$.
	For the GARCH($1, 1$) model, we fit the model via the VT estimation based on the in-sample data $\{y_t\}_{t=1}^{T_0}$, and then
	compute $\widehat{y}_{{T_0}+t_0|T_0}^2$ in the conventional way. For the LS-ARCH($q$) model, we follow the method in \citet{FSR:2008} to compute $\widehat{y}_{{T_0}+t_0|T_0}^2$. That is,
	we treat the last $\widetilde{T}$ in-sample data points as if they came from a stationary ARCH($q$) process, and then estimate the parameters based on these $\widetilde{T}$ data points
	and compute $\widehat{y}_{{T_0}+t_0|T_0}^2$ as for the stationary ARCH($q$) model.
	Here, the tuning parameter $\widetilde{T}$ is chosen by minimizing  the QLIKE, i.e.,
	\begin{equation*}
	\widetilde{T}=\arg\min_{T\in\mathcal{T}}\sum_{t\in[{T_0}-50,\,{T_0}-1]}\log\widehat{y}_{t+1|t}^2(T)+\frac{y_{t+1}^2}{\widehat{y}_{t+1|t}^2(T)},
	\end{equation*}
	where $\mathcal{T}=\{50, 100, ..., 500\}$, and $\widehat{y}_{t+1|t}^2(T)$
	computed as for the stationary ARCH($q$) model is the prediction of $y_{t+1}^2$ based on the data $\{y_{i}\}_{i=t-T+1}^{t}$.

	Table \ref{tableforecast} reports the values of QLIKE($t_0$) for all four models, where
	the prediction horizon $t_0$ is taken as $1, 5, 10$, and $22$, corresponding to
	daily, weekly, biweekly, and monthly predictions, respectively. The DM test  in \cite{DM1995} is implemented to compare the forecasting accuracy between the model with smallest value of QLIKE and other three models.	
	 From this table, we find that  for both series, the S-GARCH model has the smallest value of QLIKE for
$t_0=1$ and 5,  while  the GARCH model has the smallest value of QLIKE for
$t_0=10$ and 22. In terms of DM test, we find that
the model with smallest value of QLIKE does not exhibit significantly
forecasting accuracy than its three competitors for USD/INR series, while
it has significantly
forecasting accuracy than two ARCH-type competitors for FTSE series. These findings are
consistent with those in \citet{FSR:2008} and \citet{Truquet:2017}
that (S-)GARCH models could deliver better forecasts than LS-ARCH and S-ARCH models.
For the S-GARCH model, we simply just use the latest in-sample long-run component estimator $\widehat{\tau}_{T_0}$
to predict the out-of-sample long-run component $\tau_{T_0+t_0}$.
So far, we do not know how to find an ``optimal'' way under certain criterion to predict $\tau_{T_0+t_0}$, and this dilemma seems to exist in
most of nonparametric methods. We believe that with a better prediction of $\tau_{T_0+t_0}$,
our S-GARCH model could deliver better prediction performances especially at longer prediction horizons.

	\begin{table}[!h]
\addtolength{\tabcolsep}{-0.001pt}
		\centering
		\caption{The values of QLIKE($t_0$) for $t_0=1, 5, 10$ and 22}\label{tableforecast}
				\begin{tabular}{cccccccccc}
				\hline
				& \multicolumn{4}{c}{USD/INR}                                       &  & \multicolumn{4}{c}{FTSE}                                      \\ \cline{2-5} \cline{7-10}
				& 1              & 5              & 10             & 22             &  & 1             & 5             & 10            & 22            \\ \hline
				S-GARCH & $-\bf{1.6981}$ & -$\bf{1.6541}$ & -1.6050        & -1.5479        &  & $\bf{0.7033}$ & $\bf{0.7600}$ & $0.8179^*$        & $0.8771^*$        \\
				S-ARCH  & -1.6863        & -1.6273        & -1.6175        & -1.5641        &  & $0.7310^*$        & $0.7867^*$        & $0.8304^*$        & $0.8926^*$        \\
				GARCH   & -1.6815        & -1.6520        & -$\bf{1.6180}$ & -$\bf{1.5774}$ &  & 0.7043        & 0.7665        & $\bf{0.7982}$ & $\bf{0.8321}$ \\
				LS-ARCH & -1.6710        & -1.6212        & -1.6155        & -1.5722        &  & $0.7520^*$        & $0.7861^*$        & $0.8058^*$        & $0.8635^*$        \\ \hline
			\end{tabular}
			\begin{tablenotes}
				\item[1] {\scriptsize Note: For each $t_0$, the smallest value of QLIKE($t_0$) among all four models is in boldface.  }
				\item [2]{\scriptsize DM test is implemented between the model with smallest value of QLIKE($t_0$)  and other three models, where the symbol star (*) indicates the significance at 5\% level.}
			\end{tablenotes}

	\end{table}

\section{Concluding remarks}
	This paper provides a complete statistical inference procedure for the S-GARCH model.
	Our methodologies including the estimation and testing focus on the QMLE of non-time-varying parameters in GARCH-type
	short run component. Since this QMLE is based on the estimate of the long run component,
	we develop new proof techniques to derive its asymptotic normality, and find that
	its asymptotic variance is adaptive to the long run component with unknown form.
	By comparing the results with those in \citet{HL:2010}, we find a much simpler asymptotic variance expression for the
	QMLE, bringing the convenience of use to practitioners. By comparing with the QMLE from the VT method in \citet{FHZ:2011},
	we find that our QMLE not only enjoys a broader application scope to
	deal with the non-stationary S-GARCH model, but also avoids any efficiency loss when the S-GARCH model is stationary.
	All of these interesting features have not been unveiled before in the literature, and they make our QMLE and its related Lagrange multiplier and portmanteau tests
	more appealing in practice.

	Finally, we suggest some future research topics. First, it is interesting to extend our study to the robust estimation context.
	This could give us more efficient estimators and more powerful tests for dealing with heavy-tailed data. Second,
	a similar semiparametric framework as (\ref{semi_model}) can be posed into
	many variants of the standard GARCH model
	(e.g., the asymmetric power-GARCH model in \citet{PWT:2008} and the
	asymmetric log-GARCH model in \citet{FWZ:2013}),
	and
	our methodologies could be applied to these new resulting semiparametric models.
	Third, another possible
	 work is to relax the smooth condition of the long run component to
	allow for abrupt changes. This seems challenging and may require more non-trivial technical treatments.
	
\section*{Acknowledgments}
The authors greatly appreciate the helpful comments
and suggestions of two anonymous referees, Associate Editor,
and Co-Editor.
Jiang acknowledges that his work was partly carried out during the visit in University of Hong Kong and University of Illinois at Urbana-Champaign,
and his work is supported by China Scholarship Council  (No. 201906210093).
Li's work is supported by the NSFC (Nos. 11771239 and 71973077) and the Tsinghua University Initiative Scientific Research Program (No. 2019Z07L01009). Zhu's work is supported by Hong Kong GRF grant (Nos. 17306818 and 17305619), NSFC (Nos. 11690014 and 11731015), Seed Fund for Basic Research (No. 201811159049), and Fundamental Research Funds for the Central University (19JNYH08).

\appendix
\section*{Appendix: Proofs}

To facilitate the proofs, we first introduce some notations. As for $g_{t}(\theta)$, $\widehat{g}_{t}(\theta)$, $L_{T}(\theta)$, and $\widehat{L}_{T}(\theta)$ in (\ref{ideal_llf})--(\ref{hatg}), we similarly define
\begin{equation}\label{mle}
\widetilde{L}_T(\theta)=\sum_{t=1}^{T}\widetilde{l}_t(\theta)\quad \text{with}\quad \widetilde{l}_t(\theta)=\frac{u_t^2}{\widetilde{g}_t(\theta)}+\log\widetilde{g}_t(\theta),
\end{equation}
where $\widetilde{g}_t(\theta)$ is defined in the same way as $\widehat{g}_t(\theta)$ in (\ref{hatg}) with
$\widehat{u}_{t}$ replaced by $u_{t}$. Meanwhile, we let $\kappa_T=\sqrt{(\log T)/(Th)}+h^2$, $\Delta_t=\widehat{u}^2_t-u^2_t$, $\widetilde{S}_t(\theta)={\widehat{g}_t(\theta)}^{-1}-{\widetilde{g}_t(\theta)}^{-1}$,
$B^{(j)}$ be a $p\times p$ matrix with $(1,j)$th element 1 and other elements 0, and
$$
B=\left(\begin{matrix}
\beta_1&\beta_2&\cdots&\beta_p\\
1&0&\cdots&0\\
\vdots&&&\vdots\\
0&\cdots&1&0
\end{matrix}\right)_{p\times p}.
$$
Also, we let $C$ be a generic constant which may differ at each
appearance.

Next, we give five technical lemmas, whose proofs are given in the supplementary material (\citet{JLZ:2019}). Lemma \ref{lem_linton} captures the error from the nonparametric estimation.
Lemma \ref{lem_diff} gives some useful results on $\Delta_t$  and $\widetilde{S}_t(\theta)$.
Lemma \ref{lem_hat} ensures that replacing $u_t$ by $\widehat{u}_t$ has a negligible impact
on our asymptotic results. Lemma \ref{lem_initial}  guarantees that the effect from initial values to our asymptotics
is negligible. Lemma \ref{lem_mixing} provides a useful $\beta$-mixing result.

\begin{lemma}\label{lem_linton}
	Suppose Assumptions \ref{ident_tau}--\ref{ass_ut} hold. Then, almost surely $(a.s.)$,
	
	$\mathrm{(i)}$ ${\displaystyle \sup_{x\in(0,1)}\Big|\widehat{\tau}(x)-\tau(x)-\frac{1}{T}\sum_{t=1}^{T}K_h\Big(x-\frac{t}{T}\Big)\tau(x)(u^2_t-1)-h^2b(x)\Big|=O\Big(\frac{\log T}{Th}\Big)+o(h^2)}$;
	
	$\mathrm{(ii)}$
	${\displaystyle \sup_{x\in(0,1)}\Big|\frac{1}{T}\sum_{t=1}^{T}K_h\Big(x-\frac{t}{T}\Big)\tau(x)(u^2_t-1)\Big|=O\Big(\sqrt{\frac{\log T}{Th}}\Big)}$.
\end{lemma}


\begin{lemma}\label{lem_diff}
	
	Suppose the conditions in Theorem \ref{thm_garch} hold. Then,
	
    $(\mathrm{i})$	
	$\Delta_t={\tau_t}^{-1}{(\tau_t-\widehat{\tau}_t)u^2_t}+O(\kappa_T^2)u^2_t$, where $O(1)$ holds uniformly in $t$;
	
	$(\mathrm{ii})$	
	$\sup_{\theta\in\Theta}|\widetilde{S}_t(\theta)|\leq C\kappa_T$.
	
\end{lemma}


\begin{lemma}\label{lem_hat}
	Suppose the conditions in Theorem \ref{thm_garch} hold. Then, for any  $\iota\leq 4(1+2\delta)$,
	
$(\mathrm{i})$ ${\displaystyle \sup_{\theta\in\Theta}\Big\|\widehat{g}_t(\theta)-\widetilde{g}_t(\theta)\Big\|_{\iota}\leq C\kappa_T}$;

$(\mathrm{ii})$ ${\displaystyle \sup_{\theta\in\Theta}\Big\|\frac{\partial \widehat{g}_t(\theta)}{\partial\theta}-\frac{\partial\widetilde{g}_t(\theta)}{\partial\theta}\Big\|_{\iota}\leq C\kappa_T}$;

$(\mathrm{iii})$ ${\displaystyle \sup_{\theta\in\Theta}\Big\|\frac{\partial^2 \widehat{g}_t(\theta)}{\partial\theta\partial\theta'}-\frac{\partial^2 \widetilde{g}_t(\theta)}{\partial\theta\partial\theta'}\Big\|_{\iota}\leq  C\kappa_T}$.
\end{lemma}


\begin{lemma}\label{lem_initial}
Suppose Assumptions \ref{ident_tau} and \ref{ass_garch}--\ref{ass_eta}  hold. Then, there exists a $\rho\in(0,1)$ such that for any $\iota\leq 4(1+2\delta)$,

$(\mathrm{i})$ ${\displaystyle \sup_{\theta\in\Theta}\Big\|g_t(\theta)-\widetilde{g}_t(\theta)\Big\|_{\iota}\leq C\rho^t}$;

$(\mathrm{ii})$ ${\displaystyle \sup_{\theta\in\Theta}\Big\|\frac{\partial g_t(\theta)}{\partial\theta}-\frac{\partial\widetilde{g}_t(\theta)}{\partial\theta}\Big\|_{\iota}\leq C\rho^t}$;

$(\mathrm{iii})$ ${\displaystyle \sup_{\theta\in\Theta}\Big\|\frac{\partial^2 g_t(\theta)}{\partial^2\theta\partial\theta'}-\frac{\partial^2 \widetilde{g}_t(\theta)}{\partial\theta\partial\theta'}\Big\|_{\iota}\leq C\rho^t}$.
\end{lemma}


\begin{lemma}\label{lem_mixing}
	Suppose Assumptions  \ref{ass_ut}--\ref{ass_garch} and \ref{ass_eta}(i) hold. Then, $\big\{(u_t,g_t,\frac{\partial g_t(\theta_0)}{\partial\theta'})\big\}$ is strictly stationary and $\beta$-mixing with exponential decay.
\end{lemma}


\textsc{Proof of Theorem \ref{thm_kernel}.} See the supplementary material in \citet{JLZ:2019}.
\qed

\vspace{4mm}

\textsc{Proof of Theorem \ref{thm_garch}($\mathrm{i}$).}  See the supplementary material in \citet{JLZ:2019}.

\vspace{4mm}

In order to prove Theorem \ref{thm_garch}($\mathrm{ii}$), we need a crucial proposition, which
is interesting in its own right.

\begin{pro}\label{keypro}
	Let $\{c_t\}_{t\in\mathbb{Z}}$ be a sequence of stationary process and  $\mathcal{F}_{t}^{s}=\sigma(c_i,t\leq i\leq s)$ be the sigma-filed generated by  $\{c_i,t\leq i\leq s\}$. Define
$$S_{T}=\frac{1}{\sqrt{T}}\sum_{t=1}^{T}b_t\Big\{\frac{1}{Th}\sum_{s=1}^{T}K\Big(\frac{s-t}{Th}\Big)a_s\Big\},$$
where $a_t=f(c_t)$, $b_t=g(c_t,c_{t-k})$ for some $k\leq n_T$, and $f(\cdot)$ and $g(\cdot,\cdot)$ are two real-valued functions.
Suppose the following conditions hold:
	
	$(1)$ $Ea_t=0$, $Eb_t=0$, $E|a_t|^{\iota_1(1+2\delta)}<\infty$ and $E|b_t|^{\iota_2(1+2\delta)}<\infty$, where $\iota_1,\iota_2>0$ satisfy $\iota_1^{-1}+\iota_{2}^{-1}=1/2$ and  $\delta>0$;
	
	$(2)$ $c_t$ is $\beta$-mixing with mixing coefficients $\beta(j)$ satisfying $\sum_{j=1}^{\infty}\beta(j)^{\delta/(1+\delta)}<\infty$;
	
	$(3)$ $K(\cdot)$ satisfies Assumption \ref{ass_kernel} and $h$ satisfies Assumption \ref{ass_bandwidth};	
	
	$(4)$ $n_T$ is either a constant or $n_T\to \infty$ and $n_T=o(\sqrt{Th^2})$ as $T\to \infty$.
	
\noindent Then,
	\begin{flalign*}
	(\mathrm{i})~ |ES_t|\leq \frac{Cn_T}{\sqrt{T}h} \quad\text{and}\quad (\mathrm{ii})~ES_T^2\leq C\max\Big\{\frac{n_T}{\sqrt{Th}},\frac{1}{Th^2}\Big\}.
	\end{flalign*}
\end{pro}

\begin{proof} We decompose $S_T=S_{T,1}+S_{T,2}+S_{T,3}+S_{T,4}$, where
\begin{flalign*}
S_{T,1}=&\frac{1}{T^{3/2}h}\sum_{t=1}^{T-1}b_t\sum_{s=t+1}^{T}K\Big(\frac{s-t}{Th}\Big)a_s,\quad	 S_{T,2}=\frac{1}{T^{3/2}h}\sum_{t=n_T+1}^{T}b_t\sum_{s=1}^{t-n_T}K\Big(\frac{s-t}{Th}\Big)a_s,\\
S_{T,3}=&\frac{1}{T^{3/2}h}\sum_{t=n_T+1}^{T-1}b_t\sum_{s=t-n_T+1}^{t}K\Big(\frac{s-t}{Th}\Big)a_s,\quad
S_{T,4}=\frac{1}{T^{3/2}h}\sum_{t=1}^{n_T}b_t\sum_{s=1}^{t}K\Big(\frac{s-t}{Th}\Big)a_s.
\end{flalign*}

$(\mathrm{i})$
Under Condition (1), we have that $\iota_1>2$ and $\iota_2>2$,
which indicate $\|a_t\|_{2(1+2\delta)}<\infty$ and  $\|b_t\|_{2(1+2\delta)}<\infty$.
Since $b_t\in\mathcal{F}_{t-n_T}^{t}$, by Conditions (1)--(3) and Davydov's inequality (see \citet{Davydov:1968}), we have
\begin{flalign*}
|ES_{T,1}|\leq& \frac{C}{T^{3/2}h}\sum_{t=1}^{T-1}\sum_{s=t+1}^{T}|E(b_ta_s)|
\\\leq& \frac{C}{T^{3/2}h}\sum_{t=1}^{T-1}\sum_{s=t+1}^{T}\beta(s-t)^{\delta/(1+\delta)}\|b_t\|_{2(1+\delta)}\|a_s\|_{2(1+\delta)}=\frac{C}{T^{1/2}h}.
\end{flalign*}
Similarly, we can show that $|ES_{T,2}|\leq\frac{C}{T^{1/2}h}$. The result holds by noticing that
\begin{flalign*}
&E|S_{T,3}|\leq\frac{C}{T^{3/2}h}\sum_{t=n_T+1}^{T-1}\sum_{s=t-n_T+1}^{t}E|b_ta_s|
\leq\frac{Cn_T}{T^{1/2}h},\\
&E|S_{T,4}|\leq\frac{C}{T^{3/2}h}\sum_{t=1}^{n_T}\sum_{s=1}^{t}E|b_ta_s|
\leq\frac{Cn_T^2}{T^{3/2}h}.
\end{flalign*}

$(\mathrm{ii})$
It is not hard to obtain that $ES_{T,3}^2\leq \frac{Cn_T^2}{Th^2}$ and $ES_{T,4}^2\leq \frac{Cn_T^4}{T^3h^2}$.  Below, we only prove that $ES_{T,1}^2\leq\frac{Cn_T}{\sqrt{Th}}+\frac{C}{Th^2}$, since we can similarly show that $ES_{T,2}^2\leq\frac{C}{\sqrt{Th}}+\frac{C}{Th^2}$.

Let $\varpi_t=\sum_{s=t+1}^{T}K(\frac{t-s}{Th})a_s$. Then,
$ES_{T,1}^2=\frac{1}{T^3h^2}\sum_{t=1}^{T}\sum_{t'=1}^{T}
Eb_tb_{t'}\varpi_t\varpi_{t'}:=V_{T,1}+V_{T,2}+V_{T,3}$,
where
\begin{flalign*}
V_{T,1}=&\frac{1}{T^3h^2}\sum_{t=1}^{T}\sum_{t'=t+1}^{T}
Eb_tb_{t'}\varpi_t\varpi_{t'},\quad
V_{T,2}=\frac{1}{T^3h^2}\sum_{t=1}^{T}
E(b_t\varpi_t)^2,\\	
V_{T,3}=&\frac{1}{T^3h^2}\sum_{t'=1}^{T}\sum_{t=t'+1}^{T}
Eb_tb_{t'}\varpi_t\varpi_{t'}.
\end{flalign*}
For simplicity, we only show that $V_{T,1}\leq \frac{Cn_T}{\sqrt{Th}}+\frac{C}{Th^2}$.
Let
\begin{flalign*}
\varpi_{1t}=&\sum_{s=t+1}^{t'}K\Big(\frac{t-s}{Th}\Big)a_s\in\mathcal{F}_{t+1}^{\min\{t'-1,t+[Th]\}},\quad\varpi_{2t}=\sum_{s=t'+1}^{T}K\Big(\frac{t-s}{Th}\Big)a_s\in\mathcal{F}_{t'}^{\infty},\\
\varpi_{1t'}=&\sum_{s=t'+1}^{t'+[Th]-1}K\Big(\frac{t'-s}{Th}\Big)a_s\in\mathcal{F}_{t'}^{t'+[Th]-1},\quad\varpi_{2t'}=\sum_{s=t'+[Th]}^{T}K\Big(\frac{t'-s}{Th}\Big)a_s\in\mathcal{F}_{t'+[Th]}^{\infty}.
\end{flalign*}
Here, we have used the fact that $K(\frac{t-s}{Th})=0$ if $|t-s|>[Th]$ by Condition (3).
Moreover, decompose $V_{T,1}=V_{T,11}+V_{T,12}+V_{T,13}$,
where
\begin{flalign*}
V_{T,11}=&\frac{1}{T^3h^2}\sum_{t=1}^{T}\sum_{t'=t+1}^{T}Eb_t\varpi_{1t}b_{t'}\varpi_{1t'},\quad
V_{T,12}=\frac{1}{T^3h^2}\sum_{t=1}^{T}\sum_{t'=t+1}^{T}Eb_t\varpi_{1t}b_{t'}\varpi_{2t'},\\
V_{T,13}=&\frac{1}{T^3h^2}\sum_{t=1}^{T}\sum_{t'=t+1}^{T}Eb_t\varpi_{2t}b_{t'}\varpi_{t'}.
\end{flalign*}
Using Lemmas \ref{lem_key1}--\ref{lem_key3} below, it follows that $V_{T,1}\leq C\max\big\{\frac{n_T}{\sqrt{Th}},\frac{1}{Th^2}\big\}$.
\end{proof}


\begin{lemma}\label{lem_key1}
	Under the conditions in Proposition \ref{keypro}, $V_{T,11}\leq C\max\big\{\frac{1}{\sqrt{Th}},\frac{1}{Th^2}\big\}$.
\end{lemma}

\begin{proof}
First, we decompose $V_{T,11}=\sum_{i=1}^{3}V_{T,11i}$, where
\begin{flalign*}
V_{T,111}=&\frac{1}{T^3h^2}\sum_{t=1}^{T}\sum_{t'=t+1}^{t+2[Th]}\mathrm{Cov}(b_t\varpi_{1t}b_{t'},\varpi_{1t'}), V_{T,112}=\frac{1}{T^3h^2}\sum_{t=1}^{T}\sum_{t'=t+2[Th]+1}^{T}\mathrm{Cov}(b_t\varpi_{1t},b_{t'}\varpi_{1t'}),\\
V_{T,113}=&\frac{1}{T^3h^2}\sum_{t=1}^{T}\sum_{t'=t+2[Th]+1}^{T}E(b_t\varpi_{1t})E(b_{t'}\varpi_{1t'}).
\end{flalign*}

Next, by  Theorem 4.1 in \citet{SY:1996}, we have
\begin{flalign}
\label{varpi2}&\|\varpi_{1t}\|_{\iota}\leq C\sqrt{\min\{t'-t-1,[Th]+1\}}\|a_s\|_{\iota+\xi_0}
\end{flalign}
for some $0\leq \iota\leq \iota_1(1+2\delta)-\xi_0 $ and $\xi_0>0$.
Since $b_t\varpi_{1t}b_{t'}\in\mathcal{F}_{-\infty}^{t'}$, by Davydov's inequality  and H\"older's inequality, we can obtain
\begin{flalign*}
|V_{T,111}|\leq& \frac{1}{T^3h^2}\sum_{t=1}^{T}\sum_{t'=t+1}^{t+2[Th]}\sum_{s=t'}^{t'+[Th]-1}C\beta(s-t')^{\delta/(1+\delta)}
\|b_t\varpi_{1t}b_{t'}\|_{(1+\delta)\iota_1/(\iota_1-1)}\|a_s\|_{\iota_1(1+\delta)}\\
\leq&\frac{1}{T^3h^2}\sum_{t=1}^{T}\sum_{t'=t+1}^{t+2[Th]}\sum_{s=t'}^{t'+[Th]-1}C\beta(s-t')^{\delta/(1+\delta)}
\|b_t\|^2_{\iota_2(1+\delta)}\|\varpi_{1t}\|_{\iota_1(1+\delta)}\|a_s\|_{\iota_1(1+\delta)}\\
=&\frac{1}{T^3h^2}\sum_{t=1}^{T}\sum_{t'=t+1}^{t+2[Th]}\|\varpi_{1t}\|_{\iota_1(1+\delta)}.
\end{flalign*}
Using (\ref{varpi2}) with $\iota=\iota_1(1+\delta)$ and $\xi_0=\iota_1\delta$, it follows that $|V_{T,111}|\leq\frac{CT^2h\sqrt{Th}}{T^3h^2}=\frac{C}{\sqrt{Th}}$.

Third, we note  that $b_t\varpi_{1t}\in\mathcal{F}_{-\infty}^{t+[Th]}$ as $t'>t+[Th]$, and $b_{t'}K(\frac{t'-s}{Th})a_s\in\mathcal{F}_{t'-n_T}^{\infty}\subset\mathcal{F}_{t'-[Th]+1}^{\infty}$ as $n_T\ll[Th]$. Then, by Davydov's inequality and H\"older's inequality, we have
\begin{flalign*}
|V_{T,112}|\leq& \frac{1}{T^3h^2}\sum_{t=1}^{T}\sum_{t'=t+2[Th]+1}^{T}\sum_{s=t'+1}^{t'+[Th]-1}
\Big|\mathrm{Cov}\big(b_{t}\varpi_{1t},b_{t'}K\big((t'-s)/Th\big)a_s\big)\Big|
\\\leq& \frac{C}{T^3h^2}\sum_{t=1}^{T}\sum_{t'=t+2[Th]+1}^{T}\sum_{s=t'+1}^{t'+[Th]-1}\beta(t'-t-2[Th])^{\delta/(1+\delta)}\|b_t\varpi_{1t}\|_{2(1+\delta)}\|b_{t'}a_s\|_{2(1+\delta)}
\\\leq&\frac{C}{T^3h^2}\sum_{t=1}^{T}\sum_{t'=t+2[Th]+1}^{T}\sum_{s=t'+1}^{t'+[Th]-1}\beta(t'-t-2[Th])^{\delta/(1+\delta)}\|b_t\|^2_{\iota_2(1+\delta)}\|\varpi_{1t}\|_{\iota_1(1+\delta)}\|a_s\|_{\iota_1(1+\delta)}
\\\leq& \frac{CTh}{T^3h^2}\sum_{t=1}^{T}\|\varpi_{1t}\|_{\iota_1(1+\delta)}.
\end{flalign*}
Using (\ref{varpi2}) with $\iota=\iota_1(1+\delta)$ and $\xi_0=\iota_1\delta$, it follows that $|V_{T,112}|\leq \frac{C}{\sqrt{Th}}$.
Finally, since it is straightforward to show  that $|V_{T,113}|\leq \frac{C}{Th^2}$, the result follows.
\end{proof}


\begin{lemma}\label{lem_key2}
	Under the conditions in Proposition \ref{keypro}, $V_{T,12} \leq C/\sqrt{Th}$.
\end{lemma}

\begin{proof}
Note that $b_t\varpi_{1t}b_{t'}\in\mathcal{F}_{-\infty}^{t'}$. By Davydov's inequality, H\"older's inequality and (\ref{varpi2}), we have
\begin{flalign*}
\big|\mathrm{Cov}(b_t\varpi_{1t}b_{t'},\varpi_{2t'})\big|
\leq&\sum_{s=t'+[Th]}^{T}K\Big(\frac{t'-s}{Th}\Big)\big|\mathrm{Cov}(b_t\varpi_{1t}b_{t'},a_s)\big|
\\\leq& C\sum_{j=[Th]}^{T}\beta(j)^{\delta/(1+\delta)}\|b_t\|^2_{\iota_2(1+\delta)}\|\varpi_{1t}\|_{\iota_1(1+\delta)}\|a_s\|_{\iota_1(1+\delta)}\\
=&C\sqrt{Th}\sum_{j=[Th]}^{T}\beta(j)^{\delta/(1+\delta)}.
\end{flalign*}
By Condition (2) and the fact $Th\to\infty$ as $T\to\infty$, we have that $Th\sum_{j=[Th]}^{T}\beta(j)^{\delta/(1+\delta)}\leq C$,
which entails that $|V_{T,12}|\leq\frac{C}{\sqrt{T^3h^5}}\leq \frac{C}{\sqrt{Th}}.$
\end{proof}


\begin{lemma}\label{lem_key3}
	Under assumptions of Proposition \ref{keypro}, $V_{T,13} \leq Cn_T/\sqrt{Th}$.
\end{lemma}

\begin{proof}
Rewrite $\varpi_{t'}=\varpi_{3t'}+\varpi_{4t'}$, where $$\varpi_{3t'}=\sum_{r=t'+1}^{s}K\Big(\frac{t'-r}{Th}\Big)a_r
\mbox{ and } \varpi_{4t'}=\sum_{r=s+1}^{t'+[Th]}K\Big(\frac{t'-r}{Th}\Big)a_r.$$
Then, we can decompose $V_{T,13}=\sum_{i=1}^{4}V_{T,13i}$, where
\begin{flalign*}
V_{T,131}=&\frac{1}{T^3h^2}\sum_{t=1}^{T}\sum_{t'=t+1}^{t+[Th]}\sum_{s=t'+1}^{t+[Th]}K\Big(\frac{t-s}{Th}\Big)\mathrm{Cov}(b_tb_{t'},a_s\varpi_{3t'}),\\
V_{T,132}=&\frac{1}{T^3h^2}\sum_{t=1}^{T}\sum_{t'=t+1}^{t+[Th]}\sum_{s=t'+1}^{t+[Th]}K\Big(\frac{t-s}{Th}\Big)\mathrm{Cov}(b_tb_{t'},a_s\varpi_{4t'}),\\
V_{T,133}=&\frac{1}{T^3h^2}\sum_{t=1}^{T}\sum_{t'=t+1}^{t+[Th]}\sum_{s=t'+1}^{t+[Th]}K\Big(\frac{t-s}{Th}\Big)E(b_tb_{t'})E(a_s\varpi_{t'}),\\
V_{T,134}=&\frac{1}{T^3h^2}\sum_{t=1}^{T}\sum_{t'=t+[Th]+1}^{T}\sum_{s=t'+1}^{T}K\Big(\frac{t-s}{Th}\Big)E(b_tb_{t'}a_s\varpi_{t'}).
\end{flalign*}

First, by interchanging summations of $s$ and $r$, we have
\begin{flalign*}
V_{T,131}=\frac{C}{T^3h^2}\sum_{t=1}^{T}\sum_{t'=t+1}^{t+[Th]}\sum_{r=t'+1}^{t+[Th]}\sum_{s=r}^{t+[Th]}
K\Big(\frac{t-s}{Th}\Big)K\Big(\frac{t'-r}{Th}\Big)\mathrm{Cov}(b_tb_{t'},a_sa_r).
\end{flalign*}
Since $b_tb_{t'}\in\mathcal{F}_{-\infty}^{t'}$, by Davydov's inequality and H\"older's inequality, we can show
\begin{flalign*}
|V_{T,131}|\leq& \frac{C}{T^3h^2}\sum_{t=1}^{T}\sum_{t'=t+1}^{t+[Th]}\sum_{r=t'+1}^{t+[Th]}
\mathrm{Cov}\big(b_tb_{t'},a_r\sum_{s=r}^{t+[Th]}K\Big(\frac{t-s}{Th}\Big)a_s\big)
\\\leq&\frac{C}{T^3h^2}\sum_{t=1}^{T}\sum_{t'=t+1}^{t+[Th]}\sum_{r=t'+1}^{t+[Th]}\beta(r-t')^{\delta/(1+\delta)}\|b_tb_{t'}\|_{\iota_2(1+\delta)/2}\\
&\quad\quad\times\Big\|a_r\sum_{s=r}^{t+[Th]}K\Big(\frac{t-s}{Th}\Big)a_s\Big\|_{\iota_1(1+\delta)}
\\\leq&\frac{C}{T^3h^2}\sum_{t=k+1}^{T}\sum_{t'=t+1}^{t+[Th]}\sum_{r=t'}^{t+[Th]}\beta(r-t')^{\delta/(1+\delta)}\|b_t\|^2_{\iota_2(1+\delta)/2}\\
&\quad\quad\times\|a_r\|_{\iota_1(1+\delta)}\Big\|\sum_{s=r}^{t+[Th]}K\Big(\frac{t-s}{Th}\Big)a_s\Big\|_{\iota_1(1+\delta)}.
\end{flalign*}
By similar arguments as for (\ref{varpi2}), we have
$$\Big\|\sum_{s=r}^{t+[Th]}K\Big(\frac{t-s}{Th}\Big)a_s\Big\|_{\iota_1(1+\delta)}\leq C\sqrt{t+[Th]-r}\|a_s\|_{\iota_1(1+2\delta)}\leq C\sqrt{Th}\|a_s\|_{\iota_1(1+2\delta)},$$
and hence it follows that $|V_{T,131}|\leq \frac{CT^2h\sqrt{Th}}{T^3h^2}=\frac{C}{\sqrt{Th}}$. Similarly, $|V_{T,132}|\leq \frac{C}{\sqrt{Th}}$.

Next, we decompose $V_{T,133}=V_{T,1331}+V_{T,1332}$, where
\begin{flalign*}
V_{T,1331}=&\frac{1}{T^3h^2}\sum_{t=1}^{T}\sum_{t'=t+1}^{t+n_T}\sum_{s=t'}^{t+[Th]}K\Big(\frac{t-s}{Th}\Big)E(b_tb_{t'})E(a_s\varpi_{t'}),\\
V_{T,1332}=&\frac{1}{T^3h^2}\sum_{t=1}^{T}\sum_{t'=t+n_T+1}^{t+[Th]}\sum_{s=t'}^{t+[Th]}K\Big(\frac{t-s}{Th}\Big)E(b_tb_{t'})E(a_s\varpi_{t'}).
\end{flalign*}
It is easy to see
\begin{flalign*}
|V_{T,1331}|\leq& \frac{1}{T^3h^2}\sum_{t=1}^{T}\sum_{t'=t+1}^{t+n_T}\sum_{s=t'}^{t+[Th]}K\Big(\frac{t-s}{Th}\Big)\|b_t\|^2_{2}\|a_s\|_2\|\varpi_{t'}\|_2
\leq\frac{CT^2h}{T^3h^2}\sum_{t'=t+1}^{t+n_T}\|\varpi_{t'}\|_2
\end{flalign*}
and $\|\varpi_{t'}\|_2\leq C\sqrt{Th}\|a_s\|_{2(1+\delta)}$ by (\ref{varpi2}). So, we have
that $|V_{T,1331}|\leq Cn_T/\sqrt{Th}$. Moreover, by Davydov's inequality and H\"older's inequality, we can show
\begin{flalign*}
|V_{T,1332}|\leq&\frac{C}{T^3h^2}\sum_{t=1}^{T}\sum_{t'=t+n_T+1}^{t+[Th]}\sum_{s=t'}^{t+[Th]}\beta(t'-n_T-t)^{\delta/(1+\delta)}\|b_t\|^2_{\iota_2(1+\delta)/2}\|a_s\|_2\|\varpi_{t'}\|_2\\\leq& \frac{CT^2h\sqrt{Th}}{T^3h^2}=\frac{C}{\sqrt{Th}},
\end{flalign*}
which implies that $|V_{T,133}|\leq Cn_T/\sqrt{Th}$.

Finally, since $K\big(\frac{s-t}{Th}\big)=0$ when $s\geq t'>t+[Th]$, it follows that $V_{T,134}=0$, and hence the
result follows.
\end{proof}


\textsc{Proof of Theorem \ref{thm_garch}($\mathrm{ii}$).} By Taylor's expansion, we have
\begin{flalign}\label{taylor}
\sqrt{T}(\widehat{\theta}_T-\theta_0)=-\Big\{\frac{1}{T}\frac{\partial^2 \widehat{L}_T(\theta^*)}{\partial\theta\partial\theta'}\Big\}^{-1}
\frac{1}{\sqrt{T}}\frac{\partial \widehat{L}_T({\theta}_0)}{\partial\theta},
\end{flalign}
where $\theta^*$ lies between $\hat{\theta}_T$ and $\theta_0$.

Let $\widehat{g}_t=\widehat{g}_{t}(\theta_0), \widehat{g}_t=\widehat{g}_{t}(\theta_0)$,
$\frac{\partial \widehat{g}_t}{\partial\theta_m}=\frac{\partial \widehat{g}_t(\theta_0)}{\partial\theta_m}$,
$\frac{\partial \widetilde{g}_t}{\partial\theta_m}=\frac{\partial \widetilde{g}_t(\theta_0)}{\partial\theta_m}$,
and $\widetilde{S}_t=\widetilde{S}_t(\theta_0)$.
By noting that $\widehat{g}_t^{-1}=\widetilde{S}_t+\widetilde{g}_t^{-1}$, $\widehat{g}_t^{-1}\widehat{u}^2_t=\widetilde{g}_t^{-1}u^2_t+u^2_t\widetilde{S}_t+\widetilde{g}_t^{-1}\Delta_t$ and $\frac{\partial \widehat{g}_t}{\partial\theta_m}=\frac{\partial \widehat{g}_t}{\partial\theta_m}-\frac{\partial \widetilde{g}_t}{\partial\theta_m}+\frac{\partial \widetilde{g}_t}{\partial\theta_m}$, we have
\begin{flalign} \label{expansionlikelihood}
\frac{1}{\sqrt{T}}\frac{\partial \widehat{L}_T({\theta}_0)}{\partial\theta_m}=\sum_{i=1}^{12}U_i,
\end{flalign}
where
\begin{flalign*}
U_1=&\frac{1}{\sqrt{T}}\sum_{t=1}^{T}(1-\widetilde{g}_t^{-1}u^2_t){\widetilde{g}_t}^{-1}\frac{\partial \widetilde{g}_t}{\partial\theta_m},~\quad
U_2=\frac{1}{\sqrt{T}}\sum_{t=1}^{T}(1-\widetilde{g}_t^{-1}u^2_t){\widetilde{g}_t}^{-1}\Big(\frac{\partial \widehat{g}_t}{\partial\theta_m}-\frac{\partial \widetilde{g}_t}{\partial\theta_m}\Big),\\
U_3=&\frac{1}{\sqrt{T}}\sum_{t=1}^{T}(1-\widetilde{g}_t^{-1}u^2_t)\widetilde{S}_t\frac{\partial \widetilde{g}_t}{\partial\theta_m},\quad\quad
U_4=\frac{1}{\sqrt{T}}\sum_{t=1}^{T}(1-\widetilde{g}_t^{-1}u^2_t)\widetilde{S}_t\Big(\frac{\partial \widehat{g}_t}{\partial\theta_m}-\frac{\partial \widetilde{g}_t}{\partial\theta_m}\Big),\\
U_5=&-\frac{1}{\sqrt{T}}\sum_{t=1}^{T}\widetilde{g}_t^{-1}u^2_t\widetilde{S}_t\frac{\partial \widetilde{g}_t}{\partial\theta_m},\qquad\quad
U_6=-\frac{1}{\sqrt{T}}\sum_{t=1}^{T}\widetilde{g}_t^{-1}u^2_t\widetilde{S}_t\Big(\frac{\partial \widehat{g}_t}{\partial\theta_m}-\frac{\partial \widetilde{g}_t}{\partial\theta_m}\Big),\\
U_7=&-\frac{1}{\sqrt{T}}\sum_{t=1}^{T}u^2_t\widetilde{S}_t^2\frac{\partial \widetilde{g}_t}{\partial\theta_m},\qquad\qquad~
U_8=-\frac{1}{\sqrt{T}}\sum_{t=1}^{T}u^2_t\widetilde{S}_t^2\Big(\frac{\partial \widehat{g}_t}{\partial\theta_m}-\frac{\partial \widetilde{g}_t}{\partial\theta_m}\Big),\\
U_9=&-\frac{1}{\sqrt{T}}\sum_{t=1}^{T}\Delta_t\widehat{g}_t^{-1}\widetilde{g}_t^{-1}\frac{\partial \widetilde{g}_t}{\partial\theta_m},\qquad
U_{10}=-\frac{1}{\sqrt{T}}\sum_{t=1}^{T}\Delta_t\widehat{g}_t^{-1}\widetilde{g}_t^{-1}\Big(\frac{\partial \widehat{g}_t}{\partial\theta_m}-\frac{\partial \widetilde{g}_t}{\partial\theta_m}\Big),\\
U_{11}=&-\frac{1}{\sqrt{T}}\sum_{t=1}^{T}\Delta_t\widehat{g}_t^{-1}\widetilde{S}_t\frac{\partial \widetilde{g}_t}{\partial\theta_m},\qquad~~
U_{12}=-\frac{1}{\sqrt{T}}\sum_{t=1}^{T}\Delta_t\widehat{g}_t^{-1}\widetilde{S}_t\Big(\frac{\partial \widehat{g}_t}{\partial\theta_m}-\frac{\partial \widetilde{g}_t}{\partial\theta_m}\Big).
\end{flalign*}

Using the similar proof as for Theorem 2.2 in \citet{FZ:2004}, we can show
$$U_1=\frac{1}{\sqrt{T}}\frac{\partial{L}_T({\theta}_0)}{\partial\theta_m}+o_p(1)=-\frac{1}{\sqrt{T}}\sum_{t=1}^{T}(\eta_t^2-1)\psi_t+o_{p}(1).$$
By H\"older's inequality and Lemmas \ref{lem_diff}--\ref{lem_hat}, it is not hard to prove
that $U_{i}=o_p(1) \mbox{ for }i=4,6,7,8,10,11,12$.
Combining with the results in Lemmas \ref{pro_u3}--\ref{pro_u9} below, by (\ref{expansionlikelihood}) it follows that
\begin{flalign}\label{first_de}
\frac{1}{\sqrt{T}}\frac{\partial \widehat{L}_T({\theta}_0)}{\partial\theta_m}=-\frac{1}{\sqrt{T}}\sum_{t=1}^{T}\Big\{(\eta_t^2-1)\psi_t-\frac{\omega_0}{\gamma_0}
E\Big(\frac{1}{g_t^{2}}\frac{\partial g_t}{\partial \theta_m}\Big)z_t\Big\}+o_p(1),
\end{flalign}
where
$\omega_0=1-\sum_{i=1}^{q}\alpha_{i0}-\sum_{j=1}^{p}\beta_{j0}$ and $\gamma_0=1-\sum_{j=1}^{p}\beta_{j0}$.

Using Lemmas \ref{lem_diff}--\ref{lem_initial} and the consistency of $\widehat{\theta}_T$, it follows directly that
\begin{flalign}\label{second_de}
\frac{1}{T}\frac{\partial^2 \widehat{L}_t(\theta^*)}{\partial\theta\partial\theta'}\to_p E\Big\{\frac{\partial l_t(\theta_0)}{\partial\theta\partial\theta'}\Big\}=J_1.
\end{flalign}
Thus, by (\ref{taylor}) and (\ref{first_de})--(\ref{second_de}),  we have
\begin{flalign}\label{expression_1}
\sqrt{T}(\widehat{\theta}_T-\theta_0)=J_1^{-1}\frac{1}{\sqrt{T}}\sum_{t=1}^{T}
\Big\{(\eta_t^2-1)\psi_t-\frac{\omega_0}{\gamma_0}E\Big(\frac{\psi_t}{g_t}\Big)z_t\Big\}+o_p(1).
\end{flalign}

Following \citet{HKZ:2006},  $u_t^2$ has an ARMA representation
\begin{flalign*}
u_t^2=\omega_0+\sum_{i=1}^{\max\{p,q\}}(\alpha_{i0}+\beta_{i0})u_{t-i}^2-\sum_{j=1}^{p}\beta_{j0}g_{t-j}(\eta_{t-j}^2-1)+g_t(\eta_t^2-1)
\end{flalign*}
with the convention $\beta_{i0}=0$ if $i>p$ and $\alpha_{i0}=0$ if $i>q$.
Hence, it follows that $z_t=\sum_{i=1}^{\max\{p,q\}}(\alpha_{i0}+\beta_{i0})z_{t-i}-\sum_{j=1}^{p}\beta_{j0}g_{t-j}(\eta_{t-j}^2-1)+g_t(\eta_t^2-1)$, which entails
\begin{equation}\label{z_t}
\frac{1}{\sqrt{T}}\sum_{t=1}^{T}z_{t}=\frac{\gamma_0}{\omega_0}\frac{1}{\sqrt{T}}\sum_{t=1}^{T}g_t(\eta_t^2-1)+o_p(1).
\end{equation}
By (\ref{expression_1})--(\ref{z_t}), we have
\begin{flalign}\label{bahadur}
\sqrt{T}(\widehat{\theta}_T-\theta_0)=J_1^{-1}\frac{1}{\sqrt{T}}\sum_{t=1}^{T}(\eta_t^2-1)\Big\{\psi_t-E\Big(\frac{\psi_t}{g_t}\Big)g_t\Big\}+o_p(1).
\end{flalign}
Now, the result holds by (\ref{bahadur}), the martingale central limit theorem, and the fact that
$E(g_t\psi_{t}')=E\big(\partial g_{t}/\partial\theta_0\big)=0$.
\qed

\vspace{4mm}

In order to prove Lemmas \ref{pro_u3}--\ref{pro_u9}, we need Lemma \ref{lem_difftau} below. The proofs of
Lemmas \ref{lem_difftau}--\ref{pro_u9} are provided in the supplementary material (\citet{JLZ:2019}).

\begin{lemma}\label{lem_difftau}
	Let $m_T$ satisfy \begin{equation}\label{m_T}
	m_T=O(T^{\lambda_m}) ~\text{for some } \lambda_m>0 ~ \text{and  } \lambda_m+\lambda_h<1/2.
	\end{equation}
Then, under the conditions in Theorem \ref{thm_garch},
$$\max_{1\leq i\leq m_T}\max_{i+1\leq t\leq T}\Big|\tau_t^{-1}(\widehat{\tau}_t-\tau_t)-\tau_{t-i}^{-1}(\widehat{\tau}_{t-i}-\tau_{t-i})\Big|=o\big(1/\sqrt{T}\big) \text{  a.s.}$$
\end{lemma}

\begin{lemma}\label{pro_u3}
	Under the conditions in Theorem \ref{thm_garch}, $U_2=o_p(1)$ and $U_3=o_p(1)$.
\end{lemma}

\begin{lemma}\label{pro_u5}
	Under the conditions in Theorem \ref{thm_garch}, $U_5=-\frac{1}{\sqrt{T}}\sum_{t=1}^{T}M_mz_t+o_p(1)$, where
$M_m=E\big(\frac{1}{g_t}\frac{\partial g_t}{\partial\theta_m}\big)-\frac{\omega_0}{\gamma_0}E\big(\frac{1}{g_t^{2}}\frac{\partial g_t}{\partial\theta_m}\big).$
\end{lemma}

\begin{lemma}\label{pro_u9}
	Under the conditions in Theorem \ref{thm_garch}, $U_9=\frac{1}{\sqrt{T}}\sum_{t=1}^{T}E\Big(\frac{1}{g_t}\frac{\partial g_t}{\partial\theta_m}\Big)z_t+o_p(1)$.
\end{lemma}

%
%

\textsc{Proof of Theorem \ref{LMtest}.} See the supplementary material in \citet{JLZ:2019}.

\vspace{4mm}

\textsc{Proof of Theorem \ref{thm_port}.} Since $\overline{\widehat{\eta}^2}\to_p 1$  and $\frac{1}{T}\sum_{t=1}^{T}(\widehat{\eta}^2_t(\widehat{\theta}_T)-1)^2\to_p \kappa-1$, it suffices to consider $P_k$, where
\begin{flalign*}
P_k=&\frac{1}{\sqrt{T}}\sum_{t=k+1}^{T}\Big\{\widehat{\eta}^2_t(\widehat{\theta}_T)-1\Big\}\Big\{\widehat{\eta}^2_{t-k}(\widehat{\theta}_T)-1\Big\}\\
=&\frac{1}{\sqrt{T}}\sum_{t=k+1}^{T}\Big\{\frac{\widehat{u}^2_t}{\widehat{g}_t(\widehat{\theta}_T)}-1\Big\}
\Big\{\frac{\widehat{u}^2_{t-k}}{\widehat{g}_{t-k}(\widehat{\theta}_T)}-1\Big\}
\\=&\frac{1}{\sqrt{T}}\sum_{t=k+1}^{T}\Big\{\frac{\widehat{u}^2_t}{\widehat{g}_t(\widehat{\theta}_T)}
-\frac{u^2_t}{g_t(\theta_0)}+\frac{u^2_t}{g_t(\theta_0)}-1\Big\}
\Big\{\frac{\widehat{u}^2_{t-k}}{\widehat{g}_{t-k}(\widehat{\theta}_T)}-\frac{u^2_{t-k}}{g_{t-k}(\theta_0)}+\frac{u^2_{t-k}}{g_{t-k}(\theta_0)}-1\Big\}
\\:=&R_1+R_2+R_3+R_4,
\end{flalign*}
where
\begin{flalign*}
R_1=&\frac{1}{\sqrt{T}}\sum_{t=k+1}^{T}\Big\{\frac{u^2_t}{g_t(\theta_0)}-1\Big\}\Big\{\frac{u^2_{t-k}}{g_{t-k}(\theta_0)}-1\Big\},\\
R_2=&\frac{1}{\sqrt{T}}\sum_{t=k+1}^{T}\Big\{\frac{\widehat{u}^2_t}{\widehat{g}_t(\widehat{\theta}_T)}-\frac{u^2_t}{g_t(\theta_0)}\Big\}
\Big\{\frac{\widehat{u}^2_{t-k}}{\widehat{g}_{t-k}(\widehat{\theta}_T)}-\frac{u^2_{t-k}}{g_{t-k}(\theta_0)}\Big\},\\
R_3=&\frac{1}{\sqrt{T}}\sum_{t=k+1}^{T}\Big\{\frac{u^2_t}{g_t(\theta_0)}-1\Big\}
\Big\{\frac{\hat{u}^2_{t-k}}{\widehat{g}_{t-k}(\widehat{\theta}_T)}-\frac{u^2_{t-k}}{g_{t-k}(\theta_0)}\Big\},\\
R_4=&\frac{1}{\sqrt{T}}\sum_{t=k+1}^{T}\Big\{\frac{\widehat{u}^2_t}{\widehat{g}_t(\widehat{\theta}_T)}-\frac{u^2_t}{g_t(\theta_0)}\Big\}
\Big\{\frac{u^2_{t-k}}{g_{t-k}(\theta_0)}-1\Big\}.
\end{flalign*}

By Lemmas \ref{r2}--\ref{r4} below, we have
$$P_k=\frac{1}{\sqrt{T}}\sum_{t=k+1}^{T}(\eta_t^2-1)(\eta_{t-k}^2-1)
-D_k\sqrt{T}(\widehat{\theta}_T-\theta_0)-\frac{\omega_0}{\gamma_0}H_k\Big(\frac{1}{\sqrt{T}}\sum_{t=1}^{T}z_t\Big).$$

Together with (\ref{z_t})--(\ref{bahadur}), it follows that
\begin{flalign*}
P_k=&\frac{1}{\sqrt{T}}\sum_{t=k+1}^{T}(\eta_t^2-1)(\eta_{t-k}^2-1)
-D_kJ^{-1}\Big[\frac{1}{\sqrt{T}}\sum_{t=1}^{T}(\eta_t^2-1)\Big\{\psi_t-E\Big(\frac{\psi_t}{g_t}\Big)g_t\Big\}\Big]\\
&-H_k\Big\{\frac{1}{\sqrt{T}}\sum_{t=1}^{T}g_t(\eta_t^2-1)\Big\}+o_p(1).
\end{flalign*}
The result follows by the martingale central limit theorem. \qed


\begin{lemma}\label{r2}
	Under the conditions in Theorem \ref{thm_port},	$R_2=o_p(1).$
\end{lemma}

\begin{lemma}\label{r3}
	Under the conditions in Theorem \ref{thm_port}, $R_3=o_p(1)$.
\end{lemma}

\begin{lemma}\label{r4}
	Under the conditions in Theorem \ref{thm_port}, $R_{4}=-D_k\sqrt{T}(\widehat{\theta}_T-\theta_0)-\frac{\omega_0}{\gamma_0}H_k$ $\big(\frac{1}{\sqrt{T}}\sum_{t=1}^{T}z_t\big)+o_p(1)$.
\end{lemma}

The proofs of Lemmas \ref{r2}--\ref{r4} are given in the supplementary material (\citet{JLZ:2019}).

\vspace{4mm}

\textsc{Proof of Theorem \ref{thm_improve}.}  See the supplementary material in \citet{JLZ:2019}.

\vspace{4mm}
	
\textsc{Proof of Theorem \ref{thm_BEKK}.}  See the supplementary material in \citet{JLZ:2019}.

\end{document}